\documentclass[11pt]{article}
\usepackage[margin=0.8in]{geometry}

\usepackage{varioref}
\usepackage[usenames,svgnames,xcdraw,table]{xcolor}
\definecolor{DarkBlue}{rgb}{0.1,0.1,0.5}
\definecolor{DarkGreen}{rgb}{0.1,0.5,0.1}

\usepackage{nicefrac}

\usepackage{hyperref}
\hypersetup{
   colorlinks   = true,
 linkcolor    = DarkBlue, 
 urlcolor     = DarkBlue, 
	 citecolor    = DarkGreen 
}

\usepackage{appendix}

\usepackage{mathtools}
\usepackage{amsthm}
\usepackage{algorithmic}
\usepackage{algorithm}
\usepackage{array}

\usepackage[font={small,it}]{caption}

\usepackage{graphicx,epsfig,amsmath,latexsym,amssymb,verbatim}

\usepackage[square,semicolon,numbers]{natbib}

\newcommand{\extra}[1]{}

\usepackage{amsmath,amssymb,amsfonts}
\usepackage{amsthm}
\usepackage{thmtools,thm-restate}
\usepackage{enumitem}

\newtheorem{theorem}{Theorem}
\newtheorem{corollary}{Corollary}
\newtheorem{definition}{Definition}
\newtheorem{lemma}{Lemma}
\newtheorem{proposition}{Proposition}
\newtheorem{claim}{Claim}

\def\squareforqed{\hbox{\rlap{$\sqcap$}$\sqcup$}}
\def\qed{\ifmmode\squareforqed\else{\unskip\nobreak\hfil
\penalty50\hskip1em\null\nobreak\hfil\squareforqed
\parfillskip=0pt\finalhyphendemerits=0\endgraf}\fi}
\def\endenv{\ifmmode\;\else{\unskip\nobreak\hfil
\penalty50\hskip1em\null\nobreak\hfil\;
\parfillskip=0pt\finalhyphendemerits=0\endgraf}\fi}

\renewenvironment{proof}{\noindent \textbf{{Proof~} }}{\qed\medskip}
\newenvironment{proof+}[1]{\noindent \textbf{{Proof #1~} }}{\qed\medskip}

\mathchardef\ordinarycolon\mathcode`\:
\mathcode`\:=\string"8000
\def\vcentcolon{\mathrel{\mathop\ordinarycolon}}
\begingroup \catcode`\:=\active
  \lowercase{\endgroup
  \let :\vcentcolon
  }

\DeclareMathOperator*{\argmin}{arg\,min}
\DeclareMathOperator*{\argmax}{arg\,max}

\usepackage{physics}
\usepackage[colorinlistoftodos,prependcaption,textsize=tiny]{todonotes}

\usepackage{thm-restate,amsthm,thmtools}
\usepackage{titlesec}
\usepackage{fancyhdr}
\usepackage{verbatim}
\usepackage{tabularx}

\newcommand{\PE}{\rm PE}
\newcommand{\Mat}{\mathcal{M}}
\newcommand{\Apo}{\mathcal{A}}
\newcommand{\MMS}{\textsc{MMS}}
\newcommand{\XOS}{\textsc{XOS}}

\newcommand{\Alloc}{\mathcal{A}}
\newcommand{\Palloc}{\mathcal{P}}
\newcommand{\NW}{\textsc{NSW}}

\newcommand{\EFone}{\textsc{EF}1}
\newcommand{\PO}{\textsc{PO}}

\title{\bfseries Truthful and Fair Mechanisms for Matroid-Rank Valuations}

\author{Siddharth Barman\thanks{Indian Institute of Science. {\tt barman@iisc.ac.in}} \and Paritosh Verma\thanks{Purdue University. {\tt paritoshverma97@gmail.com}}}

\date{}

\begin{document}

\maketitle

\begin{abstract}
We study the problem of allocating indivisible goods among strategic agents. We focus on settings wherein monetary transfers are not available and each agent's private valuation is a submodular function with binary marginals, i.e., the agents' valuations are matroid-rank functions. In this setup, we establish a notable dichotomy between two of the most well-studied fairness notions in discrete fair division; specifically, between envy-freeness up to one good ($\EFone$) and maximin shares ($\MMS$).   

First, we show that a Pareto-efficient mechanism of Babaioff et al.~(2021) is group strategy-proof for finding $\EFone$ allocations, under matroid-rank valuations. The group strategy-proofness guarantee strengthens the result of Babaioff et al.~(2021), that establishes truthfulness (individually for each agent) in the same context. Our result also generalizes a work of Halpern et al.~(2020), from binary additive valuations to the matroid-rank case.  

Next, we establish that an analogous positive result cannot be achieved for $\MMS$, even when considering truthfulness on an individual level. Specifically, we prove that, for matroid-rank valuations, there does not exist a truthful mechanism that is \emph{index oblivious}, Pareto efficient, and maximin fair.   

For establishing our results, we develop a characterization of truthful mechanisms for matroid-rank functions. This characterization in fact holds for a broader class of valuations (specifically, holds for binary $\XOS$ functions) and might be of independent interest. 
\end{abstract}

\section{Introduction}
The field of discrete fair division studies the allocation of indivisible goods (i.e., goods that cannot be fractionally divided)  among agents with possibly different preferences. Such allocation problems arise naturally in many real-world settings, e.g., assignment of flats in public housing estates \cite{deng2013story}, allocating courses to students \cite{budish2017course}, or distributing computational resources. Motivated, in part, by such applications, in recent years a significant body of work has been directed towards the development of algorithms that find fair and economically efficient allocations~\cite{BCEL+16,endriss2017trends}.

In the fair division literature, {envy-freeness} is one of the most prominent fairness criterion. An allocation is said to be envy-free iff every agent values her bundle (allocated goods) at least as much as the bundle of any other agent. Notably, when the goods are indivisible, the existence of envy-free allocations is  ruled out, even in rather simple instances; consider a setting with a single indivisible good that is desired by two agents. In light of such non-existence results, 
meaningful relaxations of envy-freeness (and other classic fairness criteria) have been the focus of research in discrete fair division. In the context of indivisible goods, the two most prominent notions of fairness are {envy-freeness up to one good} ($\EFone$) \cite{LMMS04,Bud11} and {maximin shares} ($\MMS$) \cite{Bud11}.

An allocation is said to be envy-free up to one good $(\EFone)$ iff each agent values her bundle no less than the bundle of any other agent, subject to the removal of one good from the other agent's bundle. For indivisible items, $\EFone$ allocations are guaranteed to exist under monotone valuations and can be computed efficiently~\cite{LMMS04}. In fact, under additive\footnote{An agent's valuation function $v$ is said to be additive iff $v(S) = \sum_{g \in S} v(\{g\})$ for every subset of goods $S$.} valuations, there necessarily exist allocations that are both $\EFone$ and Pareto efficient (i.e., economically efficient)~\cite{CKMP+19,barman2018finding}. 

Maximin share is a threshold-based fairness notion. That is, every agent has a threshold---referred to as her \emph{maximin share}---and an allocation is deemed to be maximin fair ($\MMS$) iff in the allocation each agent receives a value at least as much as her maximin share. Conceptually, the threshold follows from executing a discrete version of the cut-and-choose protocol: among $n$ participating agents, the maximin share for an agent $i$ is defined as the maximum value that $i$ can guarantee for herself by partitioning the set of (indivisible) goods into $n$ subsets, and then receiving a minimum valued (according to $i$'s valuation) one. In contrast to envy-freeness up to one good, $\MMS$ allocations are not guaranteed to exist under additive valuations \cite{KPW16,PW14}. This fairness notion, however, admits strong approximation guarantees for additive valuations and beyond; see \cite{garg2020improved,ghodsi2018fair}, and references therein.  

Equipped with relevant fairness criteria (such as $\EFone$ and $\MMS$), research in discrete fair division has primarily focussed on existence and computational  tractability of fairness notions, along with their impact on economic efficiency.  Another key desiderata in this context---and resource-allocations settings, in general---is that of \emph{truthfulness}, i.e., one wants mechanisms wherein the participating agents cannot gain by misreporting their valuations. However, in settings wherein monetary transfers are not available, these three central objectives (of fairness, economic efficiency, and truthfulness) cannot be achieved together, even under additive valuations; it is known that, under additive valuations, the only Pareto-efficient and truthful mechanism is serial dictatorship \cite{klaus2002strategy}, which is notably unfair. 

Motivated by these considerations, an important thread of work is aimed at identifying expressive valuation classes that admit truthful and fair mechanisms (without money). Of note here are valuations with binary marginals, i.e., valuations that change by at most one upon the addition or removal of a good, from any bundle. Such dichotomous preferences have been extensively studied in fair division literature; see, e.g.,~\cite{BM04,BL16,KPS15}. Here, two well-studied function classes, in order of containment, are: binary additive valuations and binary submodular\footnote{Recall that a set function $v : 2^{[m]} \mapsto \mathbb{Z}_+$ is said to be submodular iff $v(S \cup \{g\}) - v(S) \geq v(T \cup \{g\}) - v(T)$, for all subsets $S \subseteq T \subseteq [m]$ and $g \notin T$} valuations. Binary additive (respectively, submodular) valuations are additive (submodular) functions with binary marginals. We note that binary submodular functions characteristically correspond to matroid-rank functions \cite{schrijver2003combinatorial} and, hence, constitute a combinatorially expressive function class. Matroid-rank valuations capture preferences in many resource-allocation domains; see \cite{benabbou2020finding} for pragmatic examples. 

Halpern et al.~\cite{halpern2020fair} provide a \emph{group strategy-proof} mechanism for binary additive valuations. Their mechanism is based on maximizing the Nash social welfare (i.e., the geometric mean of the agents' valuations) with a lexicographic tie-breaking rule. Nash optimal allocations are Pareto efficient. Also, under binary additive valuations, such allocations are known to be $\MMS$ as well as $\EFone$. Hence, for binary additive valuations, the work of Halpern et al.~\cite{halpern2020fair} achieves all the three desired properties; their mechanism in fact can be executed in polynomial time.  

For the broader class of binary submodular valuations, Babaioff et al.~\cite{babaioff2020fair} develop a truthful, Pareto efficient, and fair mechanism. This work considers Lorenz domination as a fairness criterion and, hence as implications, obtains $\EFone$ and $\nicefrac{1}{2}$-$\MMS$ guarantees. 

Contributing to this thread of work, the current studies mechanism design, without money, for fairly allocating indivisible goods. We focus on settings wherein the agents' valuations are matroid-rank functions (i.e., are binary submodular functions) and establish a notable dichotomy between $\EFone$ and $\MMS$. 

\paragraph{Our Results.}
First, we show that the Pareto-efficient mechanism of Babaioff et al.~\cite{babaioff2020fair} is group strategy-proof for finding $\EFone$ allocations, under matroid-rank valuations (Theorem \ref{theorem:gsp}). The group strategy-proofness guarantee strengthens the result of Babaioff et al.~\cite{babaioff2020fair}, that establishes truthfulness (individually for each agent). Our result also generalizes the work of Halpern et al.~\cite{halpern2020fair}, from binary additive valuations to the matroid-rank case.  

Next, we establish that an analogous positive result cannot be achieved for $\MMS$, even when considering truthfulness for individual agents. Specifically, we prove that, for matroid-rank valuations, there does not exist a truthful mechanism that is \emph{index oblivious}, Pareto efficient, and maximin fair (Theorem \ref{theorem-mms-negative}). For establishing our results, we develop a characterization of truthful mechanisms under matroid-rank functions (Theorem \ref{theorem_truthful_monotonic}). This characterization in fact holds for a broader class of valuations (specifically, holds for binary $\XOS$ functions) and might be of independent interest. 

\paragraph{Additional Related Work.} Barman and Verma~\cite{barman2021existence} have shown that, under matroid-rank valuations, allocations that are $\MMS$ and Pareto efficient always exist. Hence, considering the existence of truthful mechanisms in this context is a well-posed question. Recall that in contrast to the matroid-rank case, under additive valuations, (exact) $\MMS$ allocations are not guaranteed to exist \cite{KPW16,PW14}. Both positive and negative mechanism-design results have been obtained for various fairness notions under additive valuations; see \cite{amanatidis2017truthful,amanatidis2016truthful,markakis2011worst}, and references therein. These results are incomparable with the ones obtained in the current work, since they address additive valuations and, by contrast, we focus on matroid-rank functions.  
\section{Notation and Preliminaries}
\label{section:notation}
We study mechanisms, without money, for partitioning $[m]=\{1, 2, \ldots, m\}$ indivisible goods among $[n]=\{1, 2, \ldots, n\}$ agents in a fair and economically efficient manner. The cardinal preferences of the agents $i \in [n]$, over subsets of goods, are specified via valuation functions $v_i: 2^{[m]} \mapsto \mathbb{R}_+$; here, $v_i(S) \in \mathbb{R}_+$ denotes the value that agent $i \in [n]$ has for a subset of goods $S \subseteq [m]$. The valuation functions of all agents are collectively represented by a valuation profile $\vb{v} = (v_1, v_2, \ldots, v_n)$. In this setup, an instance of the fair division problem corresponds to a tuple $\langle [m], [n], \vb{v} \rangle$. Our goal is to obtain fair and economically efficient allocations. Specifically, an \emph{allocation} $\Alloc = (A_1, A_2, \ldots, A_n)$ is an $n$-partition of all the goods (i.e., $\cup_{i=1}^n A_i = [m]$ and $A_i \cap A_j = \emptyset$ for all $i \neq j$) wherein subset $A_i$ is assigned to agent $i \in [n]$. The assigned subsets will be referred to as bundles.

We will use the term \emph{partial allocation} to refer to a collection of pairwise-disjoint subsets of goods $\Palloc = (P_1, P_2, \ldots, P_n)$, in which subset $P_i$ is assigned to agent $i$. In contrast to an allocation, for a partial allocation $\mathcal{P} = (P_1, \ldots, P_n)$ we may have $\cup_{i=1}^n P_i \subsetneq [m]$, i.e., it is not necessary that all the goods are assigned among the agents. Note that, for a partial allocation $\Palloc= (P_1, \ldots, P_n)$, the set of goods $[m] \setminus \left( \bigcup_{i \in [n]} P_i \right)$ remain unallocated, and $P$ is a complete allocation iff $[m] \setminus \left( \bigcup_{i \in [n]} P_i \right)= \emptyset$.

We now define the notions of fairness and economic efficiency considered in this work. 

\paragraph{Nash social welfare and Pareto optimality.} The Nash social welfare $\NW(\cdot)$ of an (partial) allocation $\Alloc=(A_1, \allowbreak \ldots, A_n)$ is the geometric mean of the agents' values in $\Alloc$, i.e., $\NW(\Alloc) \coloneqq  \left( \Pi_{i=1}^n v_i(A_i) \right)^{1/n}$. An (partial) allocation that maximizes Nash social welfare among the set of all allocations is called a Nash optimal allocation.

Given two (partial) allocations $\Alloc = (A_1, A_2, \ldots A_n)$ and $\Alloc' = (A'_1, A'_2, \ldots A'_n)$, we say that $\Alloc$ \emph{Pareto dominates} $\Alloc'$ iff for every agent $i \in [n]$, we have $v_i(A_i) \geq v_i(A'_i)$ and this inequality is strict for at least one agent. An allocation is referred to as \emph{Pareto optimal} (or $\PO$) iff there is no other allocation that Pareto-dominates it. Throughout, we will consider Pareto optimality across all allocations. In particular, under the valuations considered in the current work, there could exist partial allocations that are Pareto optimal among all (complete) allocations. 

\paragraph{Fairness notions.} This paper studies two prominent fairness criteria: envy-freeness up to one good ($\EFone$) and maximin fairness ($\MMS$). An (partial) allocation $\Palloc= (P_1, \ldots, P_n)$ is said to be $\EFone$ iff for every pair of agents $i,j \in [n]$, with $P_j \neq \emptyset$, there exists $g \in P_j$ such that $v_i(P_i) \geq v_i(P_j  \setminus \{g \})$. 

In a fair division instance $\langle [m], [n], \vb{v} \rangle$, the maximin share of agent $i \in [n]$ is defined as  
\begin{align*}
\mu_i = \max_{(X_1, \ldots, X_n)} \min_{j \in [n]} \ v_i(X_j). 
\end{align*}
Here, the maximization is considered over all possible allocations. With these agent-specific thresholds in hand, we say that an (partial) allocation $\Palloc= (P_1, \ldots, P_n)$ is maximin fair ($\MMS$) iff $v_i(P_i) \geq \mu_i$ for all agents $i \in [n]$.

\paragraph{Truthful Mechanisms.} In the current context, a mechanism, $f(\cdot)$, is a mapping from (reported) valuation profiles $\vb{v} = (v_1, v_2, \ldots, v_n)$ to (partial) allocations $\Alloc$. That is, the mechanism asks each agent $i$ to report a valuation function $v_i$ and assigns the bundles from the (partial) allocation $(A_1, \ldots, A_n) = f(v_1, \ldots, v_n)$. We will solely address deterministic mechanisms. A key desiderata is to identify mechanisms $f$ wherein it is in the best interest of each agent to report her true valuation to $f$, i.e., no agent can gain by misreporting her valuation. This requirement is formally realized through \emph{truthfulness} (also referred to as {strategy-proofness}), and it's stronger variant, \emph{group strategy-proofness}.

A mechanism $f$ is said to be \emph{truthful} iff for each agent $i \in [n]$, any valuation profile $(v_1, v_2, \ldots, v_n)$, and any valuation $v'_i$ we have $v_i(A_i) \geq v_i(A'_i)$; where $(A_1, \ldots, A_n) = f(v_1, v_2, \ldots, v_n)$ and $(A'_1, \ldots, A'_n) = f(v_1, \ldots, v_{i-1}, v'_i, v_{i+1}, \ldots, v_n)$.  
Indeed, truthfulness ensures that that agent $i$ does not receive a higher-valued (under her true valuation) bundle by misreporting her valuation to be $v'_i$. The notion can be strengthened by considering subsets of misreporting agents, instead of a single agent. 

\begin{definition}[Group Strategy-Proofness]
A mechanism $f$ is said to be group strategy-proof iff for each subset of agents $C \subseteq [n]$ and any pair of valuation profiles $(v_1, v_2, \ldots, v_n)$ and $(v'_1, v'_2, \ldots, v'_n)$, with the property that $v_j = v'_j$ for all $j \notin C$, we necessarily have $v_i(A_i) \geq v_i(A'_i)$ for some agent $i \in C$. 
\end{definition}

This definition equivalently asserts that for a group strategy-proof mechanism $f$ there does not exist a subset of colluding agents $C$ such that \emph{all} of them gain (i.e., $v_i(A'_i) > v_i(A_i)$ for all $i \in C$) by misreporting together. 

An even more demanding notion is that of \emph{strong group strategy-proofness} that requires the nonexistence of any colluding subset of agents $C$ wherein $v_i(A'_i) \geq v_i(A_i)$ for all $i \in C$, and the inequality has to be strict for at least one agent in $C$. Notably, with this stringent notion, one cannot achieve Pareto efficiency (let alone fairness) for the valuations considered in this work \cite{babaioff2020fair,BM04}.

We will throughout say that a mechanism $f$ is Pareto efficient iff, for any given valuation profile $\vb{v}$, the mechanism outputs an (partial) allocation that is Pareto optimal (with respect to $\vb{v}$). Similarly, a mechanism is said to be $\EFone$ (respectively $\MMS$) iff, for any given valuation profile, it outputs an (partial) allocation that is $\EFone$ (respectively $\MMS$). 

The current work addresses fair division instances in which, for each agent $i \in [n]$, the valuation function $v_i$ is the rank function of a matroid $\mathcal{M}_i = \left( [m], \mathcal{I}_i \right)$. Below we define rank functions and other relevant notions from matroid theory. 

For subsets $X \subseteq [m]$ and goods $g \in [m]$, we will use the following shorthands: $X + g \coloneqq X \cup \{g \}$, $X - g \coloneqq X \setminus \{ g \}$ and, $\overline{X} \coloneqq [m] \setminus X$. Also, for notational convenience, we will write $\overline{g}$ to denote the subset $[m] \setminus \{ g \}$.

\subsection{Matroid Preliminaries}
\label{section:matroid-prelims}
A pair $([m], \mathcal{I})$ is said to be a matroid iff $\mathcal{I}$ is a nonempty collection of subsets of $[m]$ (i.e., $\mathcal{I} \subseteq 2^{[m]}$) that satisfies (a) Hereditary property: if $X \in \mathcal{I}$ and $Y \subseteq X$, then $Y \in \mathcal{I}$, and (b) Augmentation property: if $X, Y \in \mathcal{I}$ and $|Y| < |X|$, then there exists $g \in X \setminus Y$ such that $Y + g \in \mathcal{I}$. Given a matroid $\mathcal{M} = ([m], \mathcal{I})$, a subset $I \subseteq [m]$ is said to be \emph{independent} iff $I \in \mathcal{I}$.

For a matroid $\mathcal{M} = ([m], \mathcal{I})$, the rank function $r: 2^{[m]} \mapsto \mathbb{Z}_+$ captures, for each subset $X \subseteq [m]$, the cardinality of the largest independent subset contained in $X$, i.e., $r(X)  \coloneqq \max\{|I| \ : \ I \in \mathcal{I} \text{ and } I \subseteq X\}$.

Rank functions bear binary marginals: $r(X \cup \{g\}) - r(X) \in \{ 0,1 \}$, for all subsets $X \subseteq [m]$ and $g \in [m]$. Also, by definition, rank functions are nonnegative ($r(X) \geq 0$ for all $X \subseteq [m]$) and monotone ($r(Y) \leq r(X)$ for all $Y \subseteq X$). Furthermore, the following characterization is well known~\cite[Chapter~39]{schrijver2003combinatorial}: any submodular function $r$ with binary marginals is in fact a matroid-rank function. 

Note that, if agent $i$'s valuation $v_i$ is the rank function of matroid $\mathcal{M}_i = ( [m], \mathcal{I}_i )$, then, for any subset $S \subseteq [m]$, we have $ v_i(S) \leq |S|$; here, equality holds iff $S$ is an independent set in $\mathcal{M}_i$, i.e., $S \in \mathcal{I}_i$.  With this observation, we next define non-wasteful allocations and mechanisms.  

\paragraph{Non-Wasteful Mechanism.} Under valuations $v_1, \ldots, v_n$, an (partial) allocation $\Alloc = (A_1, A_2, \ldots , A_n)$ is said to be \emph{non-wasteful} iff, for each agent $i \in [n]$, the assigned bundle's value $v_i(A_i) = |A_i|$. Hence, for matroid-rank valuations, this defining condition corresponds to $A_i \in \mathcal{I}_i$, for each $i \in [n]$. Furthermore, a mechanism $f$ is called \emph{non-wasteful} if it yields non-wasteful allocations for all input valuation profiles. Note that any truthful mechanism $f$ can be converted into one that is both truthful and non-wasteful: for every profile $\vb{v}$ and $(A_1, \ldots, A_n) = f(\vb{v})$, the corresponding mechanism returns a largest-cardinality independent subset $A'_i \subseteq A_i$, for each $i \in [n]$. Indeed, $(A'_1, \ldots, A'_n) \in \mathcal{I}_1 \times \ldots \times \mathcal{I}_n$ is a non-wasteful allocation and $v_i(A'_i) = v_i(A_i)$ for all agents $i$. We will establish a stronger result (Proposition \ref{proposition:nw-wlog}) showing that non-wastefulness can be achieved with additional properties and, hence, in the relevant context it can be assumed without loss of generality. 

\paragraph{Exchange Graph and Path Augmentation.} We will use certain well-known constructs from matroid theory. In particular, \emph{exchange graphs}  and the related \emph{path augmentation} operation will be utilized while establishing the results in Section \ref{section:group_strategyproofness}.

Consider a setting wherein, for each agent $i \in [n]$, the valuation $v_i$ is the rank function of matroid $\mathcal{M}_i = ([m], \mathcal{I}_i)$. Here, given a non-wasteful (partial) allocation $\Alloc = (A_1, \allowbreak \ldots, A_n) \in \mathcal{I}_1 \times \ldots \times \mathcal{I}_n$, we define the \emph{exchange graph}, $\mathcal{G}(\mathcal{A})$, to be a directed graph where the set of vertices is $[m]$ (i.e., each vertex corresponds to a good) and there is a directed edge $(g, g')$ in the graph iff for some $i \in [n]$, the good $g \in A_i$, $g' \notin A_i$, and $A_i - g + g' \in \mathcal{I}_i$; hence, exchanging good $g$ with $g'$ in the bundle $A_i$ maintains independence (with respect to $\mathcal{I}_i$).

Now we define the \emph{path augmentation} operation. If $P =(g_1, g_2, \ldots, g_k)$ is a directed path in the exchange graph $\mathcal{G}(\mathcal{A})$, then we define bundle\footnote{Recall that the \emph{symmetric difference} of any two subsets $A$ and $B$ is defined as $A \Delta B \coloneqq (A \setminus B) \cup (B \setminus A)$.} $A_i \Delta P \coloneqq A_i \ \Delta \ \{ g_j , g_{j+1} : g_j \in A_i \}$ for all $i \in [n]$ \footnote{If path $P$ is just a single vertex, then define $A_i \Delta P \coloneqq A_i$}. Hence, $A_i \Delta P$ is obtained by exchanging goods along every edge of $P$ that goes out of the set $A_i$.

Given an agent $i \in [n]$ and an independent set $X \in \mathcal{I}_i$, we define $F_i(X)$ to be the set of goods that can be added to $X$, while still maintaining its independence, i.e., $F_i(X) \coloneqq \{ g \in [m] \setminus X : X + g \in \mathcal{I}_i \}$. 

In Section \ref{section:group_strategyproofness}, we will use the following well-known result (stated in our notation) about the augmentation operation. This lemma ensures that if the augmentation is performed along a shortest path\footnote{Following standard terminology, a shortest path between two vertex sets is a path with the fewest number of edges among all paths that connect the two vertex sets.} in the exchange graph, then independence of all bundles is maintained.

\begin{restatable}[\cite{schrijver2003combinatorial}]{lemma}{LemmaPathAugmentation}
\label{lemma:path-augmentation}
Let $\mathcal{X}' = (X'_1, \ldots, X'_n)$ be any non-wasteful (partial) allocation. Additionally, for a pair of agents $i, j \in [n]$, let $Q = (g_1, g_2, \ldots, g_t)$ be a shortest path in the exchange graph $\mathcal{G}(\mathcal{X}')$ between the vertex sets $F_i(X'_i)$ and $X'_j$ (in particular, $g_1 \in F_i(X'_i)$ and $g_t \in X'_j$). Then, for all $k \in [n] \setminus \{i, j\}$, we have $X'_k \Delta Q \in \mathcal{I}_k$, along with $\left( X'_i \Delta Q \right) + g_1 \in \mathcal{I}_i$ and $X'_j - g_t \in \mathcal{I}_j$.
\end{restatable}
As per the above lemma, if we perform augmentation along a shortest path between $F_i(X'_i)$ and $X'_j$, we get a new non-wasteful (partial) allocation in which the valuation of agent $i$ increases by one and that of $j$ decreases by one; the valuations of all other agents remain unchanged. 

\section{Characterizing Truthfulness}
This section develops a characterization (under matroid-rank valuations) of mechanisms that are truthful and non-wasteful. Recall that in the case of matroid-rank valuations, by definition, non-wasteful mechanisms---for all input valuation profiles---output allocations $(A_1, \ldots, A_n)$ comprised of independent bundles, $A_i \in \mathcal{I}_i$, for all $i \in [n]$.  

As mentioned previously, from any truthful mechanism $f$, one can obtain a value-equivalent mechanism $f'$ that is both truthful and non-wasteful. That is, given truthfulness, one can assume non-wastefulness without loss of generality. We will establish a stronger result (Proposition \ref{proposition:nw-wlog}) showing that non-wastefulness can in fact be achieved with additional properties and, hence, in relevant contexts it can be assumed without loss of generality.

We will use the following notation. Let $v:2^{[m]} \mapsto \mathbb{R}_+$ be a valuation function and $X \subseteq [m]$ a subset of goods. Then, write $v^{X} ( \cdot)$ to denote the function obtained by restricting $v$ to the subset $X$, i.e., $v^{X} (S) \coloneqq v(S \cap X)$, for each $S \subseteq [m]$. One can verify that if $v$ is a matroid-rank function, then so is $v^{X}$, for any subset $X$. 

Also, for notional convenience we will write $v^{-g}$ for $v^{[m]\setminus \{g\}}$, i.e., $v^{-g}$ is the valuation obtained by removing good $g$ from consideration. Furthermore, for a valuation profile $\vb{v} = (v_1, v_2, \ldots, v_n)$ along with agent $i \in [n]$ and function $v'_i$, we write $\left(v'_i, v_{-i} \right)$ to denote the profile wherein the valuation of agent $i$ is $v'_i$ and the valuations of all the other agents remain unchanged. 

Our characterization asserts that any non-wasteful mechanism $f$ is truthful iff it is \emph{gradual} (see Definition \ref{definition_mono} below). Intuitively, this notion captures the idea that the output of the mechanism changes ``gradually'' under specific misreports: if an agent $i$ misreports by excluding one good from her valuation (i.e., reports $v^{-g}_i$ instead of $v_i$), then the number of goods assigned to her change by at most one. Also, if bundle $A_i$ is assigned to an agent $i$, then (mis)reporting a valuation that is restricted to a superset $X \supseteq A_i$ does not change the number of goods assigned to $i$. 

\begin{definition}[Gradual Mechanism]
\label{definition_mono}
A non-wasteful mechanism $f$ is said to be gradual iff for any agent $i \in [n]$, any valuation profile $\vb{v} = (v_1, \ldots, v_n)$, and allocation $(A_1, \ldots, A_n) = f(\vb{v})$ we have \\
($C_1$): \ $0 \leq  |A_i| - |B_i| \leq 1$, for any good $g \in [m]$ and corresponding allocation $(B_1, \ldots, B_n) =  f \left(v^{-g}_i, v_{-i} \right)$, and \\
($C_2$): \ $|A_i| = |B_i|$, for any superset $X \supseteq A_i$ and corresponding allocation $(B_1, \ldots, B_n) = f \left( v^{X}_i, v_{-i} \right)$.
\end{definition}
Note that this definition imposes conditions only on the sizes of bundles allocated by the mechanism and not on the valuations per se. Also, a repeated application of condition $C_1$ gives us \\
($C^*_1$): \ $0 \leq  |A_i| - |B_i| \leq |Y|$, for any set $Y \subseteq [m]$ and corresponding allocation $(B_1, \ldots, B_n) = f(v^{[m] \setminus Y}_i, v_{-i})$. 

The following theorem is the main result of this section. 
\begin{theorem}
\label{theorem_truthful_monotonic}
Under matroid-rank valuations, a non-wasteful mechanism $f$ is truthful iff it is gradual.
\end{theorem}

The proof of Theorem \ref{theorem_truthful_monotonic} directly follows Lemmas \ref{lemma_mono_if_truthful} and \ref{lemma_truthful_if_mono} which, respectively, establish the necessity and sufficiency of gradualness.

\begin{lemma}
\label{lemma_mono_if_truthful}
Under matroid-rank valuations, every non-wasteful and truthful mechanism $f$ is gradual. 
\end{lemma}
\begin{proof}
We will show that any given non-wasteful and truthful mechanism $f$ satisfies the two conditions in Definition \ref{definition_mono} and, hence, is gradual. 

To prove the first condition, ($C_1$), consider any agent $i \in [n]$, valuation profile $\vb{v} = (v_1, \ldots , v_n)$, and allocation $(A_1, \ldots, A_n) = f(\vb{v})$. For any good $g \in [m]$, define valuation $v'_i \coloneqq v^{-g}_i$ (i.e., $v'_i (S) = v_i(S - g)$ for all subsets $S$) and let profile $\vb{v}' = (v'_i, v_{-i})$ along with allocation $(B_1, \ldots, B_n) = f( \vb{v}')$. 

The truthfulness of mechanism $f$ ensures that, if agent $i$'s valuation is $v_i$, then she cannot gain by (mis)reporting it to be $v'_i$, i.e., $v_i(A_i) \geq v_i(B_i)$. Using this inequality and the non-wastefulness of $f$ we obtain
\begin{align*}
|A_i| & = v_i(A_i) \tag{$f$ is non-wasteful} \\
& \geq v_i(B_i) \tag{$f$ is truthful} \\
& \geq v'_i(B_i) \tag{$v'_i$ is a restriction of $v_i$} \\
& = |B_i| \tag{$f$ is non-wasteful}
\end{align*}
Therefore, we have one part of the condition ($C_1$): $|A_i| - |B_i| \geq 0$.
 
The truthfulness of $f$ also guards against the setting in which agent $i$'s valuation is $v'_i$ and she (mis)reports it as $v_i$, i.e., $v'_i(B_i) \geq v'_i(A_i)$. Furthermore, the definition of $v'_i$ gives us $v'_i (A_i) = v_i(A_i - g) \geq v_i(A_i) - 1$; recall that matroid-rank functions have binary marginals. Using these bounds we get the other part of ($C_1$): 
\begin{align*}
|B_i| & = v'_i(B_i) \tag{$f$ is non-wasteful} \\
& \geq v'_i(A_i) \tag{$f$ is truthful} \\
& \geq v_i(A_i) - 1 \\
& = |A_i| - 1 \tag{$f$ is non-wasteful}
\end{align*}

These observations imply ($C_1$): \ $0 \leq |A_i| - |B_i| \leq 1$. 

For the second condition ($C_2$), consider any set of goods $X \supseteq A_i$ and define valuation $v''_i \coloneqq v^X_i$ along with profile $\vb{v}'' \coloneqq (v''_i, v_{-i})$. Also, let allocation $(A''_1, \ldots, A''_n) = f(\vb{v}'')$. Truthfulness of $f$ implies $v_i(A_i) \geq v_i(A''_i)$ and $v''_i(A''_i) \geq v''_i(A_i)$. Hence, 
\begin{align}
|A''_i| & = v''_i(A''_i) \tag{$f$ is non-wasteful} \nonumber \\
& = v_i(A''_i \cap X) \tag{by definition of $v''_i$} \nonumber \\
& \leq v_i (A''_i) \tag{$v_i$ is monotonic} \nonumber \\
& \leq v_i(A_i) \tag{$f$ is truthful} \nonumber \\
& = |A_i| \label{ineq:primetwo}
\end{align}
The last equality follows from the non-wastefulness of $f$. Complementary to inequality (\ref{ineq:primetwo}) we have: 
\begin{align}
|A_i| & = v_i(A_i) \tag{$f$ is non-wasteful} \nonumber \\
& = v_i(A_i \cap X) \tag{since $X \supseteq A_i$} \nonumber \\
& = v''_i(A_i) \tag{by definition of $v''_i$} \nonumber \\
& \leq v''_i(A''_i) \tag{$f$ is truthful} \nonumber \\
& = |A''_i| \label{ineq:primeone}
\end{align} 
The last equality follows from the non-wastefulness of $f$. Inequalities (\ref{ineq:primetwo}) and (\ref{ineq:primeone}) establish condition ($C_2$): $|A_i| = |A''_i|$.
This completes the proof. 
\end{proof}

The following lemma provides the sufficiency aspect of Theorem \ref{theorem_truthful_monotonic}. Recall that, by definition, a gradual mechanism is non-wasteful. 

\begin{lemma}
\label{lemma_truthful_if_mono}
Under matroid-rank valuations, every gradual mechanism $f$ is truthful. 
\end{lemma}
\begin{proof}
We assume, towards a contradiction, that the gradual mechanism $f$ is not truthful. In particular, there exist an agent $i \in [n]$, valuation profile $\vb{v} = (v_1, v_2, \ldots , v_n)$, and   function (misreport) $\widehat{v}_i$ such that $v_i(A_i) < v_i(\widehat{A}_i)$ for allocations $(A_1, \ldots, A_n) = f (\vb{v})$ and $(\widehat{A}_1, \ldots, \widehat{A}_n) =  f(\widehat{v}_i, v_{-i})$. Write $\widehat{\vb{v}} \coloneqq (\widehat{v}_i, v_{-i})$. 
  
We will consider a third valuation $v'_i$ and show that it can be obtained as restrictions of $\widehat{v}_i$ as well as $v_i$. Towards this, let $X \subseteq \widehat{A}_i$ denote a subset that satisfies $v_i(X) = v_i(\widehat{A}_i) = |X|$; since $v_i$ is a matroid-rank function, such a subset $X$ is guaranteed to exist. In fact, $X$ is an independent set with respect to the matroid associated with rank function $v_i$. We define the valuation $v'_i$ by restricting $\widehat{v}_i$ to set $X$; specially, $v'_i \coloneqq \widehat{v}_i^{X}$. 

Recall that the non-wastefulness of $f$ ensures $\widehat{v}_i (\widehat{A}_i) = |\widehat{A}_i|$, i.e., $\widehat{A}_i$ is independent with respect to the matroid associated with $\widehat{v}_i$. Hence, $X \subseteq \widehat{A}_i$ is also an independent set for $\widehat{v}_i$. These observations imply $v_i(X) = |X| = \widehat{v}_i(X)$. Moreover, the independence of $X$ in both the concerned matroids implies $v_i(S \cap X) = |S \cap X| = \widehat{v}_i ( S \cap X)$ for all subsets $S \subseteq [m]$. Therefore, $v_i^X  = \widehat{v}^X_i$ and, hence, we can express $v'_i$ as a restriction of $v_i$ as well; in particular, $v'_i = \widehat{v}^X_i = v^X_i$. 

Write profile $\vb{v}' \coloneqq (v'_i, v_{-i})$ and allocation $(A'_1, A'_2, \ldots, A'_n) = f(\vb{v}')$. Note that we have three profiles $\vb{v} = (v_i, v_{-i})$, $\widehat{\vb{v}} = (\widehat{v}_i, v_{-i})$, and $\vb{v}' = (v'_i, v_{-i})$, and three respective allocations $(A_1,\ldots, A_n)$, $(\widehat{A}_1, \ldots, \widehat{A}_n)$, and $(A'_1, \ldots, A'_n)$.

We will apply the conditions from Definition \ref{definition_mono} to contradict the assumed inequality $v_i(A_i) < v_i(\widehat{A}_i)$. First, note that condition ($C_2$) ensures that agent $i$ receives the same number of goods under both the following profiles: $\widehat{\vb{v}} = (\widehat{v}_i, v_{-i})$ and $(\widehat{v}^{\widehat{A}_i}_i, v_{-i})$, i.e., in profile $(\widehat{v}^{\widehat{A}_i}_i, v_{-i})$, the number of goods allocated to agent $i$ is equal to $|\widehat{A}_i|$. Since $X \subseteq \widehat{A_i}$, one can view profile $\vb{v}' = (v'_i, v_{-i})$ as one obtained by restricting $\widehat{v}^{\widehat{A}_i}_i$ further to $v'_i = \widehat{v}^X_i$. Hence, condition ($C^*_1$)---equivalently a repeated application of ($C_1$)---applied between profiles $(\widehat{v}^{\widehat{A}_i}_i, v_{-i})$ and $\vb{v}' = (v'_i, v_{-i})$ gives us 
\begin{align}
|A'_i| \geq |\widehat{A}_i| - |\widehat{A}_i \setminus X| = |X| \label{ineq:cone}
\end{align}
   
Next, we use the non-wastefulness of $f$ to obtain the complementary bound $|A'_i| \leq |X|$: the bundle $A'_i$ is obtained by executing $f$ with agent $i$'s valuation set as $\widehat{v}^X_i$. Hence, non-wastefulness ensures that $A'_i \subseteq X$. Therefore, we have $|A'_i| \leq |X|$. This bound along with inequality (\ref{ineq:cone}) gives us 
\begin{align}
|A'_i| = |X| \label{ineq:atox}
\end{align}

Finally, recall that $v'_i = v^X_i$. Hence, applying condition ($C^*_1$) between $\vb{v}= (v_i, v_{-i})$ and $\vb{v}' = (v^X_i, v_{-i})$ we obtain 
\begin{align}
|A_i| \geq |A'_i| \underset{\text{via } (\ref{ineq:atox})}{=} |X| = v_i(\widehat{A}_i)   \label{ineq:finally}
\end{align}
The last inequality follows from the definition of $X$. Inequality (\ref{ineq:finally}) provides the desired contraction and completes the proof.  
\end{proof}

Our characterization holds for a class a functions which in fact encapsulates matroid-rank valuations. We define this function class next. 

\paragraph{Binary $\XOS$ functions.} A set function $v:2^{[m]} \mapsto \mathbb{R}_+$ is said to be binary $\XOS$ iff there exists a family of subsets $\mathcal{F} \subseteq 2^{[m]}$ such that $v(S) = \max_{F \in \mathcal{F}} |F \cap S|$, for all subsets $S \subseteq [m]$. That is, under a binary $\XOS$ function $v$, for any subset $S$, the value $v(S)$ is obtained by considering the largest (cardinality wise) subset of $S$ that is contained in the defining family $\mathcal{F}$. Indeed the rank function of a matroid $\mathcal{M} = ([m], \mathcal{I})$ is binary $\XOS$ with the family $\mathcal{F} = \mathcal{I}$.

For a binary $\XOS$ function $v(\cdot)$, define the collection of subsets $\mathcal{I}' \coloneqq \{ X \subseteq [m] \ : \ v(X) = |X| \}$ and note that $\mathcal{I}'$ satisfies the hereditary property. That is, for any pair of subsets $A \subseteq B \subseteq [m]$, if the set $B \in \mathcal{I}'$, then $A \in \mathcal{I}'$ as well. However, the other defining property of matroids (i.e., the augmentation property) may not be satisfied by the set family $\mathcal{I}'$. We note that the augmentation property of matroids is never invoked in the proofs of Lemma \ref{lemma_mono_if_truthful} and Lemma \ref{lemma_truthful_if_mono} (and, hence, for Theorem \ref{theorem_truthful_monotonic}). Hence, the proofs continue to hold true under the (weaker) assumption that the valuation functions $v_i$-s are binary $\XOS$. Therefore, our characterization also holds for the broader class of binary $\XOS$ functions.

\section{Impossibility Result}
\label{section:impossibility}
This section establishes a notable separation between $\EFone$ and $\MMS$ in the current mechanism-design context. We show that, under matroid-rank valuations, truthful mechanisms (satisfying some additional, desirable properties) do not exist for maximin fairness. By contrast, $\EFone$ admits such  truthful mechanisms; in fact, for $\EFone$ we have a stronger positive result guaranteeing group strategy-proofness (Section \ref{section:group_strategyproofness}).

We begin by defining the key concepts for our impossibility result. Let $\vb{v} = (v_1, v_2, \ldots, v_n)$ be any valuation profile and $\pi : [m] \mapsto [m]$ be a permutation of the set of goods. Note that, since $\pi$ is a bijection, for each good $g$, the inverse (pre-image) $\pi^{-1}(g)$ is unique. For each agent $i \in [n]$, we define function 
\begin{align*}
v^\pi_i (S) & \coloneqq v_i \left(  \left\{ \pi^{-1}(g) \right\}_{g \in S} \right) \quad \text{for all subsets $S \subseteq [m]$}.
\end{align*} 
That is, after reindexing the goods via $\pi$, one would obtain the value---with respect to the initial valuation $v_i$---of any subset of goods by applying $v^\pi_i$;  see, Figure \ref{figure1} for an illustration. 

Intuitively, the valuation profile $\vb{v}^\pi \coloneqq (v^\pi_1, \ldots, v^\pi_n)$ represents the same set of preferences as in profile $\vb{v}$, the only difference is that the goods have been reindexed. The following definition aims to capture the idea that indexing of goods should \emph{not} influence the values that the agents' receive. We will use $\pi(S)$ to denote the set $\left\{ \pi(g) : g \in S \right\}$ and $\pi^{-1}(S)$ for $\left\{ \pi^{-1}(g) : g \in S \right\}$.  

\begin{figure}[h]
\centerline{\includegraphics[scale=0.6]{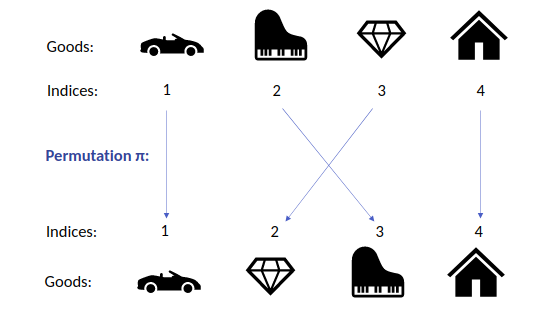}}
\caption{In this example permutation $\pi$ reindexes the goods. Then, the value $v^{\pi}_i(\{1,2\}) = v_i(\pi^{-1}(\{1,2\})) = v_i(\{1,3\})$, which is equal to the value of car and diamond.} 
\label{figure1}
\end{figure}

\begin{definition}[Index-Oblivious Mechanism]
A mechanism $f$ is said to be {index-oblivious} iff, given any valuation profile $\vb{v} = (v_1,\ldots, v_n)$ and any permutation $\pi : [m] \mapsto [m]$, for the (partial) allocations $(A_1, \ldots, \allowbreak A_n) = f(\vb{v})$ and $(A'_1, \ldots, A'_n) = f(\vb{v}^\pi)$ we have $v_i(A_i) = v^\pi_i(A'_i)$ for all agents $i \in [n]$.
\end{definition}

Note that if a mechanism is \emph{not} index-oblivious, then certain ways of indexing the goods could be advantageous for some agents and disadvantageous for others. However, indexing of goods should ideally be irrelevant. This observation supports index-obliviousness as a reasonable robustness criterion for mechanisms.     

Theorem \ref{theorem:mnwio} (proved in Appendix \ref{section:efone-index-ob}) shows for $\EFone$ there exists an index-oblivious mechanism; in particular,  we show that the Prioritized Egalitarian ($\PE$) mechanism (detailed in Section \ref{section:group_strategyproofness}) of Babaioff et al.~\cite{babaioff2020fair} is index-oblivious. It is known that, under matroid-rank valuations, $\PE$ outputs $\EFone$ allocations and it is truthful as well as Pareto efficient \cite{babaioff2020fair}.  
The following theorem proves that an analogous result is not possible for $\MMS$. 

\begin{theorem}
\label{theorem-mms-negative}
Under matroid-rank valuations, there does not exist a mechanism that is truthful, index-oblivious, Pareto efficient, and maximin fair. 
\end{theorem}
\begin{proof}
We assume, towards a contradiction, that there exists a mechanism $f$ that is truthful, index-oblivious, and it outputs Pareto efficient and $\MMS$ allocations. We can also assume that $f$ is non-wasteful (Proposition \ref{proposition:nw-wlog} in Appendix \ref{appendix:section-impossibility}). 

To derive the desired contradiction, we will construct an instance wherein an agent can always benefit by misreporting her valuation. In particular, consider a setting with $n=2$ agents and $m=6$ goods, say $\{g_1, g_2, \ldots, g_6\}$. Fix subset $G \coloneqq \{ g_1, g_2 \}$ and consider the following valuations for the two agents, respectively: $v_1(S) \coloneqq |S \cap G|$ and $v_2(S) \coloneqq \min\{1,|S\cap G|\} + \min\{2,|S\cap ([m]\setminus G)|\}$, for all subsets $S \subseteq [m]$. These valuations are rank functions of (partition) matroids. Note that under valuations $v_1$ and $v_2$, the maximin shares of the two agents are $\mu_1 = 1 $ and $\mu_2 = 3$, respectively.  

Furthermore, write allocation $(A_1, A_2) = f(v_1, v_2)$. Given the (assumed) properties of $f$, the (partial) allocation $(A_1, A_2)$ is $\MMS$. Hence, $v_1(A_1) \geq \mu_1 = 1$ and $v_2(A_2) \geq \mu_2 = 3$. By definition, the valuation $v_2$ is at most $3$, and to achieve the above-mentioned $\MMS$ bound agent $2$ must receive at least one good from $G=\{g_1, g_2\}$. The $\MMS$ guarantee for agent $1$ implies that she also receives at least one good from $G$. These observations, along with the fact that $(A_1, A_2)$ is a non-wasteful allocation ($|A_1| = v_1(A_1)$ and $|A_2| = v_2(A_2)$), ensures that  agents $1$ and $2$ receive a bundle of size $1$ and $3$, respectively.

Write $a \in G=\{g_1, g_2\}$ to denote the good assigned to agent $1$, and $b \in G$ be the other good allocated to agent $2$, i.e., $A_1 = \{a \}$ and  $A_2 = \{b, c, d\}$, for two goods $c, d \in [m]\setminus G$. Based on the three goods assigned to agent $2$, we will define three valuations profiles $\vb{w}^b$, $\vb{w}^c$, $\vb{w}^d$, and show that agent $1$ would benefit by misreporting from one at least of them. This will contradict the assumption that $f$ is truthful and, hence, establish the theorem. In all the these three profiles, we set the second agent's valuation as $w_2(S) \coloneqq |S \cap \{b, c, d\}|$, for all subsets $S \subseteq [m]$.  For each $x \in \{b, c, d \}$, let function $w^x_1 (S) \coloneqq |S \cap \{a, x \}|$ (for all subsets $S$) and write profile $\vb{w}^x = (w^x_1, w_2)$. 

A key technical step in the proof is to show that under all the three profiles $\vb{w}^x$ the two agents continue to receive the same bundles $A_1$ and $A_2$, respectively, i.e., for each $x \in \{b, c, d\}$ we have $(A_1, A_2) = f(\vb{w}^x)$. The following three claims establish this property for the profiles $\vb{w}^b$, $\vb{w}^c$, and $\vb{w}^d$, respectively. 

\begin{claim}
\label{claim:forb}
$f(\vb{w}^b) = (A_1, A_2)$.
\end{claim}
\begin{proof}
Recall that $(A_1, A_2) = f(v_1, v_2)$, with $A_1 = \{a \}$ and $A_2=\{b, c, d\}$, and $G=\{a, b\}$. Also, note that restricting $v_2$ to the subset $A_2$ in fact gives us function $w_2$ (i.e., $w_2 = v^{A_2}_2$). Since mechanism $f$ is truthful and non-wasteful, it is gradual as well (Theorem \ref{theorem_truthful_monotonic}), i.e., $f$ satisfies condition ($C_2$) in Definition \ref{definition_mono}. In particular, applying  condition ($C_2$) on agent $2$ and restricting $v_2$ to the subset $A_2$ (equivalently, considering function $w_2$) we get that agent $2$ receives three goods under allocation $(B_1, B_2) \coloneqq f(v_1, w_2)$. Hence, we have $|B_2| = |A_2| = 3$.  

Furthermore, note that $v_1 = w^b_1$; specifically, $v_1(S) = |S \cap \{a, b\}| = w^b_1(S)$, for all subsets $S$. Therefore, we have $(B_1, B_2) = f(w^b_1, w_2) = f(\vb{w}^b)$. The fact that allocation $(B_1, B_2)$ is non-wasteful implies $w_2(B_2) = |B_2|$. Given that the cardinality of $B_2$ is three and one can achieve this value under $w_2$ only if all three goods $\{b, c, d\}$ are assigned to agent $2$, we obtain $B_2 = \{b, c, d\} = A_2$.   

Additionally, bundle $B_1$ must satisfy the $\MMS$ requirement for agent $1$ (under valuation $w^b_1$), i.e., $w^b_1( B_1) \geq 1$. The definition of  $w^b_1$ and the fact that good $b$ is assigned to agent $2$ implies $B_1 = \{ a \} = A_1$.  Therefore, $(B_1, B_2) = (A_1, A_2)$ and the stated claim  holds, $(A_1, A_2) = f(\vb{w}^b)$. 
\end{proof}

\begin{claim} 
\label{claim:forc}
$f(\vb{w}^c) = (A_1, A_2)$. 
\end{claim}
\begin{proof}
To prove this claim we will invoke the index-obliviousness of mechanism $f$. Towards this, consider valuation profile $\vb{w}^b = (w^b_1, w_2)$ and permutation $\pi : [m] \mapsto [m]$ that maps $\pi(b) = c$ and $\pi(c) = b$, and $\pi(g) = g$ for each good $g \in [m] \setminus \{b,c\}$. Write valuations $\widehat{w}_1 = \left( w^b_1 \right)^\pi$ and $\widehat{w}_2 = w^\pi_2$; recall that, by definition, these valuations preserve the values under the reindexing by $\pi$. Since $f$ is index-oblivious, the allocation $(\widehat{A}_1, \widehat{A}_2) \coloneqq f(\widehat{w}_1, \widehat{w}_2)$ satisfies $\widehat{w}_1(\widehat{A}_1) = w^b_1(A_1) = 1$ and $\widehat{w}_2(\widehat{A}_2) = w_2(A_2) = 3$; here, we use Claim \ref{claim:forb}, $(A_1, A_2)=f(\vb{w}^b)$.

The definition of $\widehat{w}_2$ gives us $\widehat{w}_2 ( \widehat{A}_2) = |\widehat{A}_2 \cap \{c, b, d\}|$. Since $\widehat{w}_2(\widehat{A}_2) = 3$, it has to be the case that $\widehat{A}_2 = \{b, c, d \} = A_2$. Also, via the definition of $\widehat{w}_1$ we obtain $\widehat{w}_1 ( \widehat{A}_1) = |\widehat{A}_1 \cap \{ a, c \}|$. We also have $\widehat{w}_1(\widehat{A}_1) = 1$. These observations and the fact that good $c$ is assigned in $\widehat{A}_2$ leads to $\widehat{A}_1 = \{a\} = A_2$. Therefore, $f(\widehat{w}_1, \widehat{w}_2) = (A_1, A_2)$. 

To establish this current claim we will show that the valuation profile $\left( \widehat{w}_1, \widehat{w}_2 \right)$ is in fact the same as $\vb{w}^c = (w^c_1, w_2)$ and, hence, $f(\vb{w}^c) = f(\widehat{w}_1, \widehat{w}_2) = (A_1, A_2)$. The equivalence of valuations $\widehat{w}_1$ and $w_1^c$ follows by noting that, for every $S \subseteq [m]$, we have 
\begin{align*}
\widehat{w}_1(S) & = w_1^b(\pi^{-1}(S)) \tag{by definition of $\widehat{w}_1$}\\
& = |\pi^{-1}(S) \cap \{a,b\}| \tag{by definition of $w_1^b$}\\
& = |S \cap \{a,c\}| \tag{by definition of $\pi$}\\
& = w_1^c(S) \tag{by definition of $w^c_1$}
\end{align*}
Similarly, for each subset $S \subseteq [m]$,
\begin{align*}
\widehat{w}_2(S) & = w_2(\pi^{-1}(S)) \tag{by definition of $\widehat{w}_2$}\\
& = |\pi^{-1}(S) \cap \{b,c,d\}| \tag{by definition of $w_2$}\\
& = |S \cap \{c, b, d\}| \tag{by definition of $\pi$}\\
& = w_2(S) \tag{by definition of $w_2$}
\end{align*}
Therefore, the profile $(\widehat{w}_1, \widehat{w}_2)$ is same as $\vb{w}^c$, and the claim follows, $f(\vb{w}^c) = (A_1, A_2)$. 
\end{proof}

\begin{claim}
$f(\vb{w}^d) = (A_1, A_2)$. 
\end{claim}
\begin{proof}
This follow from an argument similar to one used for Claim \ref{claim:forc}. Here, we use the index-obliviousness of $f$ with  valuation profile $\vb{w}^b$ and permutation $\pi' : [m] \mapsto [m]$ defined as $\pi'(b) = d$, $\pi'(d) = b$, and $\pi'(g) = g$ for each $g \in [m] \setminus \{b,d\}$.
\end{proof}

We now complete the proof by showing that agent $1$ would benefit by misreporting in at least one of these three profiles. In particular, define the valuation (misreport) $w^*_1(S) \coloneqq |S \cap \{a, b, c, d \}|$, for all subsets $S \subseteq [m]$. Under the valuation $w^*_1$, the maximin share of agent $1$ is equal to two. Hence, given valuation profile $(w^*_1, w_2)$, mechanism $f$ (to maintain maximin fairness) must assign agent $1$ a bundle of value $2$. That is, for allocation $(A^*_1, A^*_2) = f(w^*_1, w_2)$, we have $w^*_1(A^*_1) \geq 2$. Furthermore, since the returned (partial) allocation must be Pareto optimal and good $a$ is only desired by agent $1$, good $a$ must be allocated to the first agent, $a \in A^*_1$. These observation imply that $A^*_1$ additionally contains at least one good from the set $\{b,c,d\}$; write $y \in A^*_1 \cap \{b, c, d\}$ to denote that good. Now, consider the case wherein agent $1$'s true valuation is $w^y_1$ (and that of agent $2$ is $w_2$). Since $y \in \{b, c, d \}$, the above-mentioned claims imply that $(A_1, A_2) = f(w^y_1, w_2)$, i.e., reporting $w^y_1$ truthfully agent $1$ receives a bundle of size $|A_1| = 1$. However, misreporting her valuation to be $w^*_1$, agent $1$ receives both the goods $a$ and $y$ (since $a, y \in A^*_1$). Therefore, $w^y_1 ( A^*_1) > w^y_1(A_1)$, and this contradicts the truthfulness of $f$. The theorem stands proved. 
\end{proof}

We note that, the above negative result can be strengthened by replacing the Pareto efficiency requirement with a weaker economic efficiency notion, which we refer to as \emph{local-efficiency}. We say that a (partial) allocation $\Alloc = (A_1, A_2, \ldots, A_n)$ is \emph{locally-efficient} iff for each good $g \in [m] \setminus \cup_{i=1}^n A_i$ and agent $i \in [n]$, we have $v_i(A_i + g) = v_i(A_i)$. Essentially, in a locally-efficient allocation we cannot increase the utility of any agent by simply allocating an unallocated good to her. Indeed, local-efficiency is a strictly weaker notion than Pareto efficiency.

In the proof of Theorem \ref{theorem-mms-negative}, Pareto efficiency of $f$ is only used (once) at the end of the proof (for arriving at the final contradiction), and there we can easily use local-efficiency instead of the stronger Pareto efficiency notion. Consequently, Theorem \ref{theorem-mms-negative} can be strengthened to obtain the following corollary.

\begin{corollary}
Under matroid-rank valuations, there does not exist a mechanism that is truthful, index-oblivious, locally-efficient, and maximin fair.
\end{corollary}

\section{Group Strategy-Proofness for $\EFone$} 
\label{section:group_strategyproofness}

This section establishes group strategy-proofness for a mechanism of Babaioff et al.~\cite{babaioff2020fair}, called the prioritized egalitarian ($\PE$) mechanism. This mechanism relies on finding \emph{Lorenz dominating} allocations; we define this notion next. 

 For any allocation $\mathcal{A}$ and valuation profile $\vb{v} = (v_1, \ldots, v_n)$, write $\vb{s}_\mathcal{A}=(s_1, s_2, \ldots, s_n)$ to denote the vector wherein all the components of  $(v_1(A_1), v_2(A_2) \ldots, v_n(A_n))$ appear in non-decreasing order, i.e., $s_1$ denotes the lowest valuation across the agents, $s_2$ is the second lowest, and so on. We say that allocation $\mathcal{A}$ Lorenz dominates another allocation $\mathcal{A}'$ iff, for every index $j \in [n]$, the sum of the first $j$ components of $\vb{s}_{\Alloc}$ is at least as large as the sum of the first $j$ components of $\vb{s}_{\Alloc'}$, i.e., iff $\vb{s}_{\Alloc}$ majorizes $\vb{s}_{\Alloc'}$ (see \cite{marshall1994inequalities} for a detailed treatment of majorization). An (partial) allocation $\Alloc$ is said to be \emph{Lorenz dominating} iff $\Alloc$ Lorenz dominates all other allocations. Notably, under matroid-rank valuations, Lorenz dominating allocations are guaranteed to exist~\cite{babaioff2020fair}.

\floatname{algorithm}{Mechanism}
\begin{algorithm}[H]
  \caption{Prioritized Egalitarian ($\PE$) \cite{babaioff2020fair}} \label{algo:pmms}
\textbf{Input:} Valuation profile $(v_1, v_2, \ldots, v_n)$ comprised of the reported valuations of all agents.\\
\noindent
\textbf{Output:} A non-wasteful Lorenz dominating allocation $\mathcal{A} = (A_1, A_2, \ldots, A_n)$.
  \begin{algorithmic}[1]
		\STATE If for any agent $i \in [n]$, the reported valuation, $v_i$, is not a matroid rank function, then reject the report and replace it with a function that is identically $0$ (which is a matroid-rank function).
		\STATE For the resulting profile $(v_1, v_2, \ldots, v_n)$, compute a non-wasteful Lorenz dominating allocation $\Alloc = (A_1, A_2, \ldots, A_n)$, breaking ties in favor of agents with lower indices. Equivalently, among all (non-wasteful) Lorenz dominating allocations, select one, $(A_1, A_2, \ldots, A_n)$, that lexicographically maximizes the vector $(v_1(A_1), v_2(A_2), \ldots, v_n(A_n))$ (i.e., the Lorenz dominating allocation maximizes $v_1(A_1)$, and then subject to that it maximizes $v_2(A_2)$, and so on).
		\STATE \textbf{return } $\Alloc = (A_1, A_2, \ldots, A_n)$
	\end{algorithmic}
\end{algorithm}
Babaioff et al.~\cite{babaioff2020fair} show that $\PE$ can be executed in polynomial time, in particular when the matroid-rank functions admit a succinct representation.

We will use a basic result from matroid theory called the strong basis exchange lemma.

\begin{lemma}[Strong Basis Exchange]
\label{lemma:basisX}
Let $M = (S, \mathcal{I})$ be a matroid and $A,B \in \mathcal{I}$ be two independent sets such that $|A| = |B|$. Then, for each element $a \in A \setminus B$, there exists $b \in B \setminus A$, such that both $A - a + b \in \mathcal{I}$ and $B - b + a \in \mathcal{I}$.
\end{lemma}

We will next prove a supporting lemma, which essentially shows that an exchange matching between two independent sets can be extended in a consistent manner. 
\begin{figure}[h]
\centerline{\includegraphics[scale=0.6]{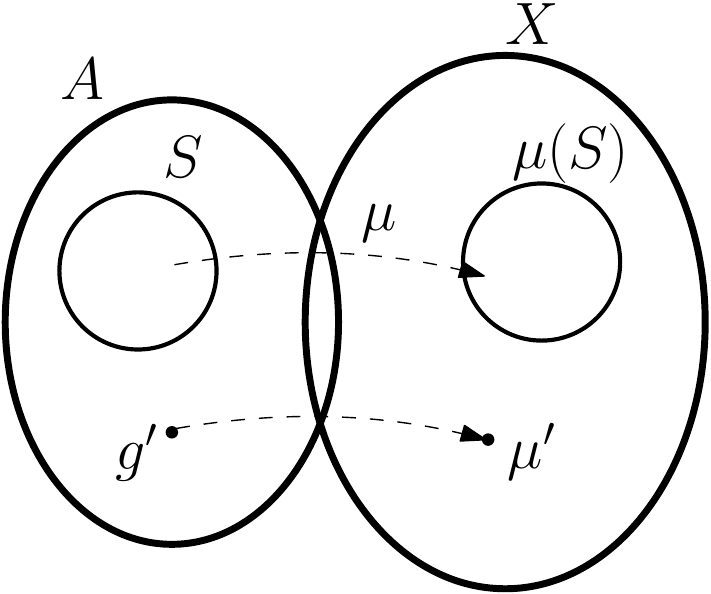}}
\caption{An illustration of Lemma \ref{lemma1}: element-wise swaps (i.e., replacing $\mu(g)$ by $g$) maintain independence in $X$. Subset swap (i.e., replacing $S$ by $\mu(S)$) maintains independence in $A$.} 
\label{figure2}
\end{figure}

\begin{lemma}
\label{lemma1}
Let $\Mat = ([m], \mathcal{I})$ be a matroid with independent sets $A,X \in \mathcal{I}$ such that $|X| \geq |A|$. Additionally, assume that for a subset of goods $S \subseteq A \setminus X$, there exists a one-to-one map $\mu:S \mapsto X \setminus A$ with the properties that (i) $X - \mu(g) + g  \in \mathcal{I}$ for each $g \in S$, and (ii) $A - S + \mu(S) \in \mathcal{I}$. 
Then, for each good $g' \in A \setminus (S + X)$, there exists a good $\mu' \in X \setminus (\mu(S) + A)$ that satisfies (i) $X - \mu' +  g' \in \mathcal{I}$ and (ii) $A - (S+g')+ (\mu(S)+\mu') \in \mathcal{I}$.
\end{lemma}
\begin{proof}
Write $A' \coloneqq A - S + \mu(S) \in \mathcal{I}$ and note that $|A'| = |A| \leq |X|$. The augmentation property of matroids implies that there exists a subset $X' \subseteq X \setminus A'$ such that $X'+ A' \in \mathcal{I}$ and $|X'+A'| = |X|$. 
Hence, we can apply strong basis exchange (Lemma \ref{lemma:basisX}) on independent sets $X$ and $X'+ A'$, to obtain that for each good $g' \in (X'+A')\setminus X$ there exists a good $\mu' \in X \setminus (A'+X')$ with the property that both the sets $A' + X' - g' + \mu'$ and $X - \mu' + g'$ are independent; the independence of the latter set proves requirement (i) of the lemma. Also, given that $A' + X' - g' + \mu' \in \mathcal{I}$ and $A' = A - S + \mu(S)$, we get $A + X' - (S+g')+ (\mu(S)+\mu') \in \mathcal{I}$. Hence, the following subset is also independent $A - (S+g')+ (\mu(S)+\mu') \in \mathcal{I}$, proving requirement (ii) of the lemma. 
We will complete the proof by showing that the good $g'$  belongs to $A \setminus (S + X)$ and $\mu'$ belongs to $X \setminus (\mu(S) + A)$. By construction, the good $g' \in (A'+X')\setminus X$, with $A' = A - S + \mu(S)$ and the subsets $\mu(S), X' \subseteq X$. These observations imply $g' \in A \setminus (S + X)$. Furthermore, good $\mu' \in X \setminus (A'+X')$ and $X \cap S = \emptyset$, hence $\mu' \in X \setminus (A + \mu(S) + X') \subset X \setminus (A + \mu(S))$. The lemma stands proved.
\end{proof}

Using essentially the same arguments as in the proof of Lemma \ref{lemma1}, one can establish the following symmetric version of the result. 

\begin{corollary}
\label{corollary1}
Let $\Mat = ([m], \mathcal{I})$ be a matroid with independent sets $A,X \in \mathcal{I}$ such that $|X| < |A|$. Also, let $Y \subseteq A \setminus X$ be a subset of goods that extends $X$ into an independent set of cardinality equal to $A$ (i.e., $X + Y \in \mathcal{I}$ and $|X + Y| = |A|$). Additionally, assume that for a subset of goods $S \subseteq A \setminus (X+Y)$, there exists a one-to-one map $\mu \ : \ S \mapsto X \setminus A$ with the properties that (i) $X - \mu(g) + g  \in \mathcal{I}$ for each $g \in S$, and (ii) $A - S + \mu(S) \in \mathcal{I}$. 
Then, for each good $g' \in A \setminus (S + X + Y)$, there exists a good $\mu' \in X \setminus (\mu(S) + A)$ such that (i) $X - \mu' +  g' \in \mathcal{I}$ and (ii) $A - (S+g')+ (\mu(S)+\mu') \in \mathcal{I}$.
\end{corollary}
For matroid $\Mat = ([m], \mathcal{I})$ and independent set $X \in \mathcal{I}$, write $F(X) \coloneqq \{ g \in [m] \setminus X : X + g \in \mathcal{I} \}$ and note that in the previous corollary, we have $Y \subseteq F(X)$.

Using the above-mentioned results for matroids, we will now prove the key technical lemma of this section. Here, we use the notation mentioned in Section \ref{section:matroid-prelims}; in particular, for any agent $i \in [n]$ and independent set $A_i \in \mathcal{I}_i$, we write $F_i(A_i) \coloneqq \{ g \in [m] \setminus A_i : A_i + g \in \mathcal{I}_i \}$. Also, given any two allocations $\Alloc=(A_1, \ldots, A_n)$ and $\mathcal{X}=(X_1, \ldots, X_n)$, define the subsets of agents $L( \mathcal{X}, \Alloc) \coloneqq \{ i \in [n] \ : \ |X_i| < |A_i| \}$ and $H(\mathcal{X}, \Alloc) \coloneqq \{ i \in [n] \ : \ |X_i| > |A_i| \}$.

\begin{lemma}
\label{lemma:revpath}
Let $\mathcal{X} = (X_1, \ldots , X_n)$ be a non-wasteful allocation and $\Alloc = (A_1, \ldots , A_n)$ be a Pareto-efficient non-wasteful allocation such that there exists an agent $h \in H(\mathcal{X}, \Alloc)$. Then, there exists a simple directed path $P = (g_k, g_{k-1}, \ldots, g_2, g_1)$ in the exchange graph $\mathcal{G}(\mathcal{X})$ that satisfies the following two properties
\begin{enumerate}
\item For the path $P$, the source vertex $g_k \in A_\ell \cap F_\ell(X_\ell)$, for some agent $\ell \in L(\mathcal{X}, \Alloc)$ and its sink vertex, $g_1 \in X_h$ (with $h \in H(\mathcal{X}, \Alloc)$).
\item For each agent $i \in [n] \setminus \{h\}$, the subsets $A'_i \coloneqq A_i \Delta \{g_{j+1},g_j : g_j \in A_i \} \in \mathcal{I}_i$ and $A'_h \coloneqq A_h \Delta \{g_{j+1},g_j : g_j \in A_h \} + g_1 \in \mathcal{I}_h$.\footnote{In contrast to the augmentation $\Delta$ considered in Lemma \ref{lemma:path-augmentation}, here we swap along edges that end in $A_i$-s. Also, note the indexing of vertices along path $P$.}
\end{enumerate}
\end{lemma}
Lemma \ref{lemma:revpath} essentially implies that both the given allocations $\mathcal{X}$ and $\mathcal{A}$ admit augmentations such that 
\begin{itemize}
\item Starting from $\mathcal{X}$, we obtain a new (non-wasteful) allocation wherein agent $\ell$'s value increases by one and that of $h$ decreases by one; all other agents continue to receive the same value (equivalently, bundle size); see Lemma \ref{lemma:path-augmentation}. 
\item Complementarily, starting from $\mathcal{A}$ we obtain a new (non-wasteful) allocation in which agent $\ell$'s value decreases by one, that of $h$ increases by one, and all other agents continue to receive the same value. 
\end{itemize}
See Figure \ref{figure3} for an illustration of the augmentations described above. The existence of such two new (non-wasteful) allocations will be crucially used in the proof of Theorem \ref{theorem:gsp}. 
\begin{figure}[h]
\centerline{\includegraphics[scale=0.5]{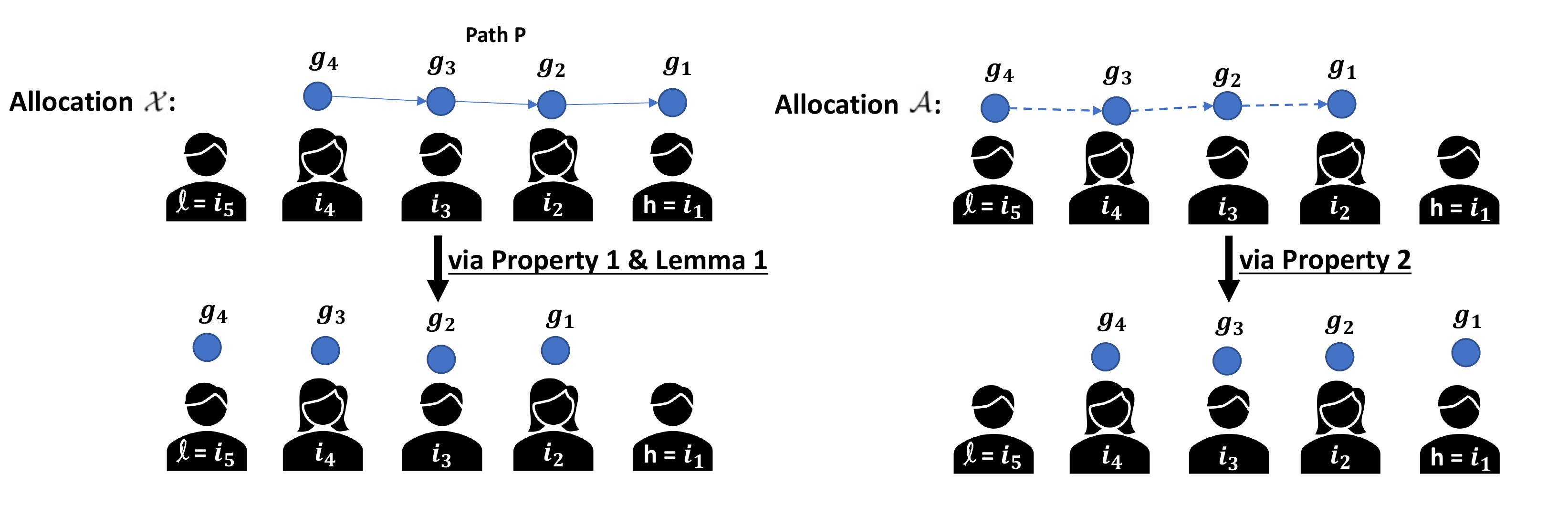}}
\caption{This illustration depicts how the two properties satisfied by path $P$ (in Lemma \ref{lemma:revpath}) are used to obtain two new allocations. For the ease of presentation, we assume that $P$ is a shortest path (between $F_\ell(X_\ell)$ and $X_h$ in $\mathcal{G}(\mathcal{X})$); Lemma \ref{lemma:path-augmentation} can be applied using any shortest path. Also, for simplicity, only the goods in $P$ are shown here.} 
\label{figure3}
\end{figure}

\begin{proof}(of Lemma \ref{lemma:revpath})
We will construct path $P$ inductively, by selecting an appropriate sink vertex $g_1$ and then adding one vertex at a time (in reverse) until we arrive at a source vertex $g_k$ such that the resulting path satisfies both conditions $1$ and $2$ (as listed in the lemma statement). At each inductive step, in order to find the next vertex, we will invoke either Lemma \ref{lemma1} or its symmetric variant, Corollary \ref{corollary1}. Furthermore, in each inductive step, a new vertex (not present in the current, partially-constructed path) is used to extend the path. This ensures that the constructed path is simple (vertex disjoint) and that the process will eventually terminate, since the number of vertices (goods) are finite. We will now describe this inductive process in detail.\\
\noindent
\emph{\bf Base case:} Start by initializing the sink of path $P$, vertex $g_1$, to be any good in $X_h \cap F_h(A_h)$; the existence of such a good follows from the augmentation property of matroids, since $|X_h| > |A_h|$ (recall that $h \in H(\mathcal{X},\Alloc)$) and $X_h, A_h \in \mathcal{I}_h$ (allocations $\mathcal{X}$ and $\Alloc$ are non-wasteful). 
Next, we initialize $i_1 \coloneqq h$ and assume that the good $g_1 \in A_{i_2}$ for some $i_2 \in [n]$. This assumption is valid because good $g_1 \in F_h(A_h)$ and as a result, if the good $g_1$ is unallocated in $\Alloc$, then $g_1$ could be allocated to agent $i_1$, thereby contradicting the Pareto-optimality of $\Alloc$. Additionally, the fact that good $g_1 \in A_{i_2}$ and $g_1 \in F_{i_1}(A_{i_1})$ (since $i_1 = h$) implies that the agents $i_2$ and $i_1$ are different, i.e., $i_2 \neq i_1$.\\
Now, suppose that the agent $i_2 \in L(\mathcal{X}, \Alloc)$ and the good $g_1 \in F_{i_2}(X_{i_2})$. In this case, with $\ell = i_2$, the path $P = (g_1)$ would satisfy condition $1$ of the lemma; since $g_1 \in A_\ell \cap F_\ell(X_\ell)$ and $g_1 \in X_h$. Additionally, condition $2$ will also be satisfied because $A'_h = A_h \Delta \{g_{j+1},g_j : g_j \in A_h \} + g_1 = A_h + g_1 \in \mathcal{I}_h$ ($g_1 \in F_h(A_h)$) and for each agent $i \in [n] \setminus \{h\}$, the subsets $A'_i = A_i \Delta \{g_{j+1},g_j : g_j \in A_i \} = A_i \in \mathcal{I}_i$ ($\Alloc$ is non-wasteful). Thus, if the assumption that the agent $i_2 \in L(\mathcal{X}, \Alloc)$ and $g_1 \in F_{i_2}(X_{i_2})$ is true, then the required path $P$ exists. Otherwise, we will show that we can always find a new vertex to extend the path $P$ while maintaining condition $2$. Towards establishing this, we consider two exhaustive cases -\\

\noindent
\emph{Case {\rm I}:} Agent $i_2 \notin L(\mathcal{X},\Alloc)$, i.e., $|X_{i_2}| \geq |A_{i_2}|$. Recall that, $g_1 \in A_{i_2} \setminus X_{i_2}$ (follows from the fact that $g_1 \in A_{i_2} \cap X_{i_1}$ and $i_1 \neq i_2$). Since $g_1 \in A_{i_2} \setminus X_{i_2}$, we can invoke Lemma \ref{lemma1} (with $S = \emptyset$, $A = A_{i_2}$, and $X = X_{i_2}$) to infer the existence of a good $g_2 \coloneqq \mu' \in X_{i_2} \setminus A_{i_2}$ such that $(a)$ $X_{i_2} - g_2 + g_1 \in \mathcal{I}_{i_2}$ and $(b)$ $A_{i_2} - g_1 + g_2 \in \mathcal{I}_{i_2}$. Note that, $(a)$ and the fact that $g_2 \in X_{i_2}$ implies that the edge $(g_2, g_1)$ lies in the graph $\mathcal{G}(\mathcal{X})$. 

We can also show that the path $P = (g_2, g_1)$ satisfies condition $2$ of the lemma statement. Since $g_1 \in A_{i_2}$, we know that $A'_{i_2} = A_{i_2} \Delta \{g_{j+1},g_j : g_j \in A_{i_2} \} = A_{i_2} - g_1 + g_2 \in \mathcal{I}_{i_2}$; here, the final containment follows from $(b)$. Furthermore, $A'_h = A_h \Delta \{g_{j+1},g_j : g_j \in A_h \} + g_1 = A_h + g_1 \in \mathcal{I}_h$, the last step follows from $g_1 \in F_h(A_h)$; for other agents $i \in [n] \setminus \{h, i_2\}$, from the definition of path augmentation we have, $A'_i = A_i \Delta \{g_{j+1},g_j : g_j \in A_i \} = A_i \in \mathcal{I}_i$. Thus, in this case, path $P$ can be extended while satisfying condition $2$. \\

\noindent
\emph{Case {\rm II}:} Agent $i_2 \in L( \mathcal{X}, \Alloc)$ (i.e., $|X_{i_2}| < |A_{i_2}|$) and $g_1 \notin F_{i_2}(X_{i_2})$. In this case, we can apply Corollary \ref{corollary1} by setting $S = \emptyset$, $X = X_{i_2}$ and $A = A_{i_2}$. Note that, $g_1 \in  A_{i_2} \setminus (X_{i_2} \cup F_{i_2}(X_{i_2}))$, implies that the requirement for invoking Corollary \ref{corollary1} (specifically, $g \in A \setminus (X \cup Y)$) is also being met, because as per the corollary statement, the set $Y$ would be a subset of $F_{i_2}(X_{i_2})$. Thus, using Corollary \ref{corollary1} we can infer the existence of a good $g_2 \coloneqq \mu' \in X_{i_2} \setminus A_{i_2}$ such that---analogous to \emph{Case {\rm I}}---condition $2$ is satisfied by the path $P = (g_2, g_1)$.

Now, we will inductively show that, if the current source of path $P$, say good $g_t$, is such that condition $1$ is not satisfied, then we can always find a new vertex $g_{t+1}$ (which is not present in $P$) to extend path $P$, maintaining condition $2$ in the process. Since there are finite number of vertices (goods), this would imply that the inductive process would eventually terminate and we would obtain the required path $P$ that satisfies both conditions $1$ and $2$. \\

\noindent
\emph{\bf Induction step:} Let $P = (g_t, g_{t-1}, \ldots, g_1)$ be the current path where good $g_t \in X_{i_t} \setminus A_{i_t}$ and (inductively) assume that $P$ satisfies condition $2$.
Assume that the good $g_t \in A_{i_{t+1}}$ for some agent $i_{t+1} \in [n]$. Such an agent must exist, because if good $g_t$ is unallocated in $\Alloc$, then the bundles $(A'_1, A'_2, \ldots, A'_n)$ form a non-wasteful allocation such that $|A'_i| = |A_i|$ for each $i \in [n] \setminus \{h\}$ and $|A'_h| = |A_h| + 1$, i.e., the Pareto-efficiency of $\Alloc$ is contradicted. Furthermore, good $g_{t} \notin X_{i_{t+1}}$ because $g_t \in X_{i_t}$ and $i_{t+1} \neq i_t$; here the last proposition follows from the fact that $g_t \in A_{i_{t+1}}$ but $g_t \notin A_{i_t}$. Hence, good $g_t \in A_{i_{t+1}} \setminus X_{i_{t+1}}$.
Suppose if the agent $i_{t+1} \in L(\mathcal{X}, \Alloc)$ and the good $g_t \in F_{i_{t+1}}(X_{i_{t+1}})$, then condition $(2)$ would be satisfied by the path $P = (g_t, g_{t-1}, \ldots, g_1)$; note that, as per the induction hypothesis $P$ already satisfies condition $(1)$. In this case, current path $P$ would be the desired path. Otherwise, similar to the base case of induction, we show that we can extend path $P$ by adding a new vertex. To show this, we consider the following cases.\\

\noindent
\emph{Case {\rm I}:} $i_{t+1} \notin L(\mathcal{X}, \Alloc)$ (i.e., $|X_{i_{t+1}}| \geq |A_{i_{t+1}}|$) and $i_{t+1} \neq h$. Here, we can apply Lemma \ref{lemma1} by using good $g' = g_t$ (note that $g_t \in A_{i_{t+1}} \setminus X_{i_{t+1}}$, satisfying the requirement for invoking the lemma), $X = X_{i_{t+1}}$, $A = A_{i_{t+1}}$, and $S = A_{i_{t+1}} \cap \{ g_{t-1}, \ldots , g_1 \}$ and the one-to-one mapping $\mu = \{ (g_k, g_{k+1}) \ : \ g_k \in A_{i_{t+1}} \text{ and } (g_{k+1}, g_k) \in P \}$. Also, for the invocation we use agent $i_{t+1}$'s matroid.
Using Lemma \ref{lemma1}, we obtain good $g_{t+1} \coloneqq \mu' \in X_{i_{t+1}} \setminus A_{i_{t+1}}$ such that $(a)$ the set $X_{i_{t+1}} - g_{t+1} + g_t \in \mathcal{I}_{i_{t+1}}$ is independent, i.e., the edge $(g_{t+1}, g_t)$ is in the graph $\mathcal{G}(\mathcal{X})$ and $(b)$ the set $A_{i_{t+1}} - (S + g_t) + (\mu(S) + g_{t+1}) \in \mathcal{I}_{i_{t+1}}$. Also, as per Lemma \ref{lemma1}, the good $g_{t+1}$ obtained is such that, $g_{t+1} \notin \mu(S)$, i.e., the good $g_{t+1}$ has previously not been included in the path $P$. We update the path $P = (g_{t+1}, g_t, \ldots , g_1)$. 

Note that, condition $(b)$ can be stated as $A_{i_{t+1}} \Delta \{g_{j+1},g_j : g_j \in A_{i_{t+1}} \} \in \mathcal{I}_{i_{t+1}}$, and since path $P$, before the update, satisfied condition $2$, we have $A_h \Delta \{g_{j+1},g_j : g_j \in A_h \} + g_1 \in \mathcal{I}_h$ along with $A_i \Delta \{g_{j+1},g_j : g_j \in A_i \} \in \mathcal{I}_i$, for each $i \in [n] \setminus \{h, i_{t+1} \}$. Therefore, after updating $P$ condition $2$ continues to be satisfied.\\

\noindent
\emph{Case {\rm II}:} $i_{t+1} = h$. Note that, in this case $i_{t+1} = h \notin L(\mathcal{X}, \Alloc)$ since $h \in H(\mathcal{X}, \Alloc)$. However, this case needs to be considered separately from \emph{Case {\rm I}}, because there the new good obtained, $g_{t+1}$, would lie in the set $X_{i_{t+1}}$ or $X_h$ since $i_{t+1} = h$. Therefore, it might happen that the new good $g_{t+1} = g_1$ (since $g_1 \in X_h$); in which case, on updating $P$ with $g_{t+1}$ ($=g_1$), $P$ would become a cycle. To avoid this scenario, we invoke Lemma \ref{lemma1} by setting $A = A_{i_{t+1}} + g_1$ ($= A_{h} + g_1$) and using the values of other parameters as used in $\emph{Case {\rm I}}$. Note that, $|X_h| > |A_h|$ implies that $|X_h| \geq |A_h + g_1|$, therefore, the requirement for invoking Lemma \ref{lemma1} ($|X| \geq |A|$) would still be met. Invoking Lemma \ref{lemma1} by setting $A = A_{h} + g_1$, ensures that the new good $g_{t+1} \notin (A_h + g_1)$ or simply $g_{t+1} \neq g_1$. Therefore, similar to the previous case, path $P$ can be updated while maintaining condition $2$. \\

\noindent
\emph{Case {\rm III}:} agent $i_{t+1} \in L(\mathcal{X}, \Alloc)$ (i.e., $|X_{i_{t+1}}| < |A_{i_{t+1}}|$) and $g_t \notin F_{i_{t+1}}(X_{i_{t+1}})$. In this case, we can apply Corollary \ref{corollary1} (instead of Lemma \ref{lemma1}) to add a new vertex to path $P$ while maintaining condition $2$.\\

This process of extending $P$, one vertex at a time, will eventually terminate, since at each step a new vertex is added and there are finite number of vertices (goods). Thus, we will eventually end up with path $P$ that satisfies both conditions $1$ and $2$.
\end{proof}

Before presenting the main theorem of this section, we will state another lemma which follows from prior work. For completeness, we provide a proof of the following lemma in Appendix \ref{appendix:MNWisLorenz}. 

\begin{restatable}{lemma}{LemmaMNWisLorenz}
\label{lemma:MNWisLorenz}
Under matroid-rank valuations, a non-wasteful (partial) allocation $\Alloc = (A_1, \ldots, A_n)$ maximizes Nash social welfare iff $\Alloc$ is Lorenz dominating.
\end{restatable}
\begin{theorem}
\label{theorem:gsp}
The $\PE$ mechanism is group strategy-proof.
\end{theorem}
\begin{proof}
Assume, towards a contradiction, that the mechanism $\PE$ is not group strategy-proof. Specifically, let $\vb{v} = (v_1, v_2, \ldots, v_n)$ be a valuation profile wherein a subset of agents $C \subseteq [n]$ can benefit  by misreporting to profile $\vb{v'} = (v'_1, v'_2, \ldots, v'_n)$; here $v'_j = v_j$ for all agents $j \notin C$. Also, write allocations $\Alloc = (A_1, A_2, \ldots, A_n) = \PE(\vb{v})$ and $\mathcal{X} = (X_1, X_2, \ldots, X_n) = \PE(\vb{v'})$. Since all the misreporting agents $i \in C$ gain by misreporting, we have $v_i(X_i) > v_i(A_i)$, for each $i \in C$.

Let $\mathcal{X}' = (X'_1, X'_2, \ldots, X'_n)$ be the non-wasteful (partial) allocation within $\mathcal{X}$, i.e., for each agent $i \in [n]$, select bundle $X'_i \coloneqq \argmax_{S \subseteq X_i} \{|S| \ : \ S \in \mathcal{I}_i \}$ and note that  $v_i(X_i) = v_i(X'_i) = |X'_i|$. Also, write $B \subseteq C$ to denote the subset of misreporting agents that receive the smallest-size bundle, $B \coloneqq \argmin_{i \in C} |X_i|$.  Considering agents in $B$, we write $h$ to denote the one with the smallest index, i.e., agent $h$ has the lowest value for $|X_h|$, among all agents $i \in C$, and subject to that, she has the smallest index. 

The remainder of the proof will be from the perspective of the valuation profile $\vb{v}$, unless stated otherwise. Note that, the allocation $\mathcal{X}'$ is non-wasteful and allocation $\Alloc$ is a Pareto efficient, since it is Lorenz dominating. Furthermore, for agent $h \in C$, we have $|X'_h| = v_h(X_h) > v_h(A_h) = |A_h|$;  equivalently $h \in H(\mathcal{X}', \Alloc)$. Hence, Lemma \ref{lemma:revpath} ensures the existence of a path $P = (g_k, \ldots , g_1)$ in $\mathcal{G}(\mathcal{X}')$ that satisfies conditions $1$ and $2$ mentioned in the lemma statement. From condition $1$, we know that path $P$ starts at $F_\ell(X'_\ell)$ (for some agent $\ell \in L(\mathcal{X}', \Apo)$) and ends at $X'_h$. 

For establishing the theorem, i.e.,  to arrive at a contradiction, we will prove two properties, $\mathcal{P}_1$ and $\mathcal{P}_2$, that directly contradict each other:
\begin{align*}
\text{Property } \mathcal{P}_1: \begin{cases} 
|A_h| < |A_\ell| - 1, & \text{ if } h > \ell  \\ 
|A_h| \leq |A_\ell| - 1, & \text{ otherwise, if } h<\ell. 
\end{cases} \\
\text{Property } \mathcal{P}_2: \begin{cases} 
|A_h| \geq |A_\ell| - 1, & \text{ if } h>\ell \\
|A_h| > |A_\ell| - 1, & \text{ otherwise, if } h<\ell
\end{cases}
\end{align*}
We will derive $\mathcal{P}_1$ by using condition $1$ (from Lemma \ref{lemma:revpath}) satisfied by path $P$ and Lemma \ref{lemma:path-augmentation}. Property $\mathcal{P}_2$ will be obtained by the fact that $P$ satisfies condition $2$ in Lemma \ref{lemma:revpath}. Hence, we complete the proof by establishing $\mathcal{P}_1$ and $\mathcal{P}_2$ next. \\

\noindent
\emph{\bf{Property $\mathcal{P}_1$:}} Recall that path $P$ starts at $F_\ell(X'_\ell)$, for some agent $\ell \in L(\mathcal{X}', \Apo)$ and ends at $X'_h$ with $h \in H(\mathcal{X}', \Apo)$. 
Since all the agents in $C$ gain (by misreporting), $C \subseteq H(\mathcal{X}', \Apo)$. Therefore, the fact that $\ell \in L(\mathcal{X}', \Apo)$ gives us $\ell \notin C$. Now, write $Q$ to denote a shortest path from $F_\ell(X'_\ell)$ to $\cup_{i \in C} X'_i$; the existence of $P$ guarantees that such a path exists. Further, assume that $Q$ ends at $X'_b$ for some $b \in C$. From the definition of agent $h$, we have\footnote{In fact, $|X_h| \leq |X_b|$  irrespective of the agents' indices.}
\begin{align}
\label{equation:1}
|X_h| \leq |X_b| \text{ if } h \leq b \text{ and } |X_h| < |X_b| \text{ if } h > b
\end{align}

Furthermore, using the facts that $|X'_h| \geq |A_h| + 1$ (since $h \in H(\mathcal{X}', \Apo)$) and $|X_h| \geq |X'_h|$ (since $X'_h \subseteq X_h$), equation (\ref{equation:1})  reduces to 
\begin{align}
\label{equation:4}
|A_h| + 1 \leq |X_b| \text{ if } h \leq b \text{, and } |A_h| + 2 \leq |X_b| \text{ if } h > b
\end{align}

Note that the path $Q$ in $\mathcal{G}(\mathcal{X}')$ is such that only its sink vertex lies in $\cup_{k \in C} X'_k$; all the other vertices in $Q$ are present in bundles $X'_j$ with $j \notin C$. For all agents $j \notin C$, the valuation functions $v'_j$ and $v_j$ are the same. Hence, the (non-wasteful) bundles $X'_j$ and $X_j$ are equal as well, for all $j \notin C$. These observations imply that the path $Q$ also lies in the exchange graph $\mathcal{G}(\mathcal{X})$, where the graph is constructed with respect to matroids corresponding to the profile $\vb{v'}$. Therefore, allocation $\mathcal{X}=\PE(\vb{v'})$ can be augmented with path $Q$ in $\mathcal{G}(\mathcal{X})$, which starts at $F_\ell(X_\ell)$ and ends at $X_b$ (see Lemma \ref{lemma:path-augmentation}).

Indeed, performing path augmentation on $\mathcal{X}$, via $Q$, will increase $|X_\ell|$ by one, decrease $|X_b|$ by one, and the bundle sizes (and values) of other agents will remain unchanged. Also, since allocation $\mathcal{X}$ is optimal with respect to $\PE$'s selection criteria (and considering profile $\vb{v}'$), the distinct allocation obtained by path augmentation must be sub-optimal (again, with respect to $\PE$'s criteria). The Lorenz domination of $\mathcal{X}$ (equivalently its Nash optimality) ensures that $|X_b| \leq |X_\ell|+1$; otherwise the resultant allocation will have higher Nash social welfare (under $\vb{v}'$). In fact, if $b > \ell$, then we must have $|X_b| \leq |X_\ell|$. Otherwise (i.e., in case $b > \ell$ and $|X_b| = |X_\ell|+1$) , the resultant allocation will have the same Nash social welfare as $\mathcal{X}$ (i.e., the resultant allocation will also be Lorenz dominating) and would get preferred (over $\mathcal{X}$) under the lexicographic tie-breaking of $\PE$.  Therefore, 
\begin{align}
\label{equation:2}
|X_b| \leq |X_\ell| \text{ if } b>\ell \text{, and } |X_b| - 1 \leq |X_\ell| \text{ if } b<\ell
\end{align}
Using the bounds $|X'_\ell| \leq |A_\ell| - 1$ (since $\ell \in L(\mathcal{X}', \Apo)$) and $|X_\ell| = |X'_\ell|$ (recall that $\ell \notin C$ and, hence, $v_\ell = v'_\ell$) along with equation (\ref{equation:2}), we obtain
\begin{align}
\label{equation:3}
|X_b| \leq |A_\ell| - 1 \text{ if } b>\ell \text{, and } |X_b| \leq |A_\ell| \text{ if } b<\ell
\end{align}

Towards establishing property $\mathcal{P}_1$, we combine equations (\ref{equation:4}) and (\ref{equation:3}) by considering the following four cases
\begin{align*}
& \text{ \emph{Case {\rm I}.} } h \leq b \text{ and } b>\ell \text{: } |A_h| + 1 \leq |X_b| \leq |A_\ell| - 1\\
& \text{ \emph{Case {\rm II}.} } h > b \text{ and } b<\ell \text{: } |A_h| + 2 \leq |X_b| \leq |A_\ell|\\
& \text{ \emph{Case {\rm III}.} } h > b \text{ and } b>\ell \text{: } |A_h| + 2 \leq |X_b| \leq |A_\ell| - 1\\
& \text{ \emph{Case {\rm IV}.} } h \leq b \text{ and } b<\ell \text{: } |A_h| + 1 \leq |X_b| \leq |A_\ell|
\end{align*}
Finally, we simplify the four cases above to obtain $\mathcal{P}_1$. Note that $h> \ell$ is possible only in \emph{Cases {\rm I}, {\rm II},} or \emph{{\rm III}}, and in all these cases we have $|A_h| + 2 \leq |A_\ell|$ or equivalently $|A_h| < |A_\ell| - 1$. Similarly, $h < \ell$ can happen only in \emph{Cases {\rm I}, {\rm II},} or \emph{{\rm IV}} and there we have $|A_h| + 1 \leq |A_\ell|$, which is same as $|A_h| \leq |A_\ell| - 1$. Therefore, we obtain property $\mathcal{P}_1$:

\begin{align}
|A_h| < |A_\ell| - 1 \text{ if } h>\ell \text{, and } |A_h| \leq |A_\ell| - 1 \text{ if } h<\ell
\end{align}

\noindent
\emph{\bf{Property $\mathcal{P}_2$:}} Here, we will use the fact that condition $2$ in Lemma \ref{lemma:revpath} is satisfied by the path $P$. The condition  implies that the bundles $\widehat{A}_h \coloneqq A_h \Delta \{g_{j+1},g_j : g_j \in A_h \} + g_1 \in \mathcal{I}_h$ and $\widehat{A}_\ell \coloneqq A_\ell \Delta \{g_{j+1},g_j : g_j \in A_\ell \} - g_k \in \mathcal{I}_\ell$ along with $\widehat{A}_i \coloneqq A_i \Delta \{g_{j+1},g_j : g_j \in A_i \} \in \mathcal{I}_i$ for each $i \notin \{h,\ell\}$. Furthermore, by definition, the bundles $\widehat{A}_1, \widehat{A}_2, \ldots, \widehat{A}_n$ are pairwise disjoint, i.e., $\widehat{\Alloc} = (\widehat{A}_1, \widehat{A}_2, \ldots, \widehat{A}_n)$ forms a non-wasteful allocation (with respect to profile $\vb{v}$).

Recall that $\Apo = \PE(\vb{v})$, and we have $|\widehat{A}_h| = |A_h|+1$ along with $|\widehat{A}_\ell| = |A_\ell|-1$; the bundle sizes of all the other agents remain unchanged. Furthermore, in contrast to $\Apo$, the distinct allocation $\widehat{A}$ must be sub-optimal under the selection criteria of $\PE$ (applied to valuation profile $\vb{v}$). Therefore, $|A_\ell| \leq |A_h| + 1$; otherwise $\widehat{\Alloc}$ will have higher Nash social welfare (under $\vb{v}$) than $\Alloc$, contradicting the optimality of $\Alloc$ (see Lemma \ref{lemma:MNWisLorenz}). In fact, if $h < \ell$, then we must have $|A_\ell| < |A_h| + 1$. Otherwise (i.e., in case $h < \ell$ and $|A_\ell| = |A_h| + 1$), the allocation $\widehat{\Alloc}$ will have the same Nash social welfare as $\mathcal{A}$ (i.e., $\widehat{\Alloc}$ will also be Lorenz dominating) and would get preferred (over $\mathcal{A}$) under the lexicographic tie-breaking of $\PE$. These observations lead to property $\mathcal{P}_2$: 
\begin{align}
|A_h| \geq |A_\ell| - 1 \text{ if } h>\ell \text{, and } |A_h| > |A_\ell| - 1 \text{ if } h<\ell
\end{align}
This completes the proof, since properties $\mathcal{P}_1$ and $\mathcal{P}_2$ directly contradict each other.
\end{proof}

\bibliographystyle{alpha} 
\bibliography{fairness,ultimate}

\appendix
\section{Missing Proofs from Section~\ref{section:impossibility}}
\label{appendix:section-impossibility}

Two mechanisms $f$ and $f'$ are said to be \emph{equivalent} iff, for every valuation profile $\vb{v} = (v_1,\ldots, v_n)$, each agent receives the same value under both $f$ and $f'$, i.e., the allocations $(A_1, \ldots, A_n) = f(\vb{v})$ and $(B_1, B_2, \ldots, B_n) = f'(\vb{v})$ satisfy $v_i(A_i) = v_i(B_i)$ for all agents $i \in [n]$.

\begin{proposition}
\label{proposition:nw-wlog}
Under matroid-rank valuations, for any truthful and index-oblivious mechanism $f$, there exists an equivalent mechanism $f'$ which, in addition to being truthful and index-oblivious, is also non-wasteful. 
\end{proposition}
\begin{proof}
We construct mechanism $f'$ from the given mechanism $f$: for any profile $\vb{v} = (v_1,\ldots, v_n)$, consisting of matroid-rank valuations, let allocation $(A_1,\ldots, A_n) = f(\vb{v})$. Now, for every bundle $A_i$, consider a subset $B_i \subseteq A_i$ with the property that  
\begin{align}
\label{equation1}
v_i(B_i) = |B_i| = v_i(A_i),
\end{align}
Since $v_i$ is a matroid-rank, such a subset $B_i$ is guaranteed to exist; in particular, $B_i$ is the largest-cardinality, independent subset of $A_i$.  If there are multiple such $B_i$s, we break ties in an arbitrary, but consistent, manner. 

For the profile $\vb{v} = (v_1,\ldots, v_n)$ and mechanism $f'$, we set the output allocation $ f'(\vb{v}) = (B_1, \ldots, B_n)$. By construction, the mechanism $f'$ is non-wasteful since $v_i(B_i) = |B_i|$, for each $i \in [n]$.

We first note that $f'$ and $f$ are equivalent, i.e., they assign the same value to each agent. This directly follows from the construction of $f'$ to satisfy $v_i(B_i) = v_i(A_i)$, for any profile $\vb{v} = (v_1,\ldots, v_n)$ and corresponding allocations  $(A_1,\ldots, A_n) = f(\vb{v})$ and $(B_1,\ldots, B_n) = f'(\vb{v})$. 

Furthermore, the equivalence of $f$ and $f'$ ensures that truthfulness and index-obliviousness for the latter mechanism. Specifically, since $f$ is truthful, for each agent $i \in [n]$, any profile $\vb{v} = (v_1,\ldots, v_n)$, and function (misreport) $v'_i$, we have
\begin{align}
\label{equation2}
v_i(A_i) \geq v_i(A'_i).
\end{align}
Here, allocations $(A_1,\ldots, A_n) = f(\vb{v})$ and $(A'_1, A'_2, \ldots, A'_n) = f(v_1, \ldots v_{i-1}, v'_i, v_{i+1}, \ldots, v_{n})$. Towards establishing truthfulness of $f'$, consider allocations $(B_1,\ldots, B_n) = f'(\vb{v})$ and $(B'_1, B'_2, \ldots, B'_n) = f'(v_1, \ldots, \allowbreak v_{i-1}, v'_i, v_{i+1}, \ldots, v_{n})$. By definition of $f'$, we know that, the subset $B'_i \subseteq A'_i$ and, hence, $v_i(B'_i) \leq v_i(A'_i)$. This inequality and equation (\ref{equation2}) gives us 
\begin{align*}
v_i(B'_i) & \leq v_i(A'_i) \\
& \leq v_i(A_i) \tag{via (\ref{equation2})} \\
& = v_i(B_i) \tag{via (\ref{equation1})}.
\end{align*}
The obtained bound, $v_i(B_i) \geq v_i(B'_i)$, implies that the mechanism $f'$ is truthful. 

Finally, we show that $f'$ is index-oblivious as well. For any permutation $\pi : [m] \mapsto [m]$, write allocation $(A^\pi_1, A^\pi_2, \ldots, A^\pi_n) = f(v^\pi_1, v^\pi_2, \ldots, v^\pi_n)$ and $(B^\pi_1, B^\pi_2, \ldots, B^\pi_n) = f'(v^\pi_1, v^\pi_2, \ldots, v^\pi_n)$. Since the mechanism $f$ is index-oblivious, for all agents $i \in [n]$, we have
\begin{align}
\label{equation3}
v_i(A_i) = v^\pi_i(A^\pi_i),
\end{align}
In addition, the construction of mechanism $f'$ ensures, for all agents $i \in [n]$: 

\begin{align}
\label{equation4}
v_i(A_i) = v_i(B_i) \text{ and } v^\pi_i(A^\pi_i) = v^\pi_i(B^\pi_i).
\end{align}
Combining these equalities we get
\begin{align*}
v_i(B_i) & = v_i(A_i) \tag{via (\ref{equation4})} \\
& = v^\pi_i(A^\pi_i) \tag{via (\ref{equation3})} \\
& = v^\pi_i(B^\pi_i) \tag{via (\ref{equation4})}.
\end{align*}
The obtained equality, $v_i(B_i) = v^\pi_i(B^\pi_i)$ implies that $f'$ is index-oblivious. This completes the proof.  
\end{proof}

\section{Nash Welfare and Lorenz Domination}
\label{appendix:MNWisLorenz}

This section provides a proof of Lemma \ref{lemma:MNWisLorenz}. We begin by stating the relevant definitions and known results that will lead to the proof. 

For a valuation profile $(v_1, \ldots, v_n)$ and an (partial) allocation $\Alloc = (A_1, A_2, \ldots, A_n)$, write $\vb{s}_{\Alloc} = (s_1, s_2, \ldots, s_n)$ to denote the sorted---in non-decreasing order---version of vector $(v_1(A_1), v_2(A_2) \ldots, v_n(A_n))$, i.e., $s_1$ denotes the lowest valuation across the agents, $s_2$ is the second lowest, and so on. 

\begin{definition}[Leximin Allocation]
An (partial) allocation $\Alloc = (A_1, A_2, \ldots, A_n)$ is said to be \emph{leximin} iff, among the set of all (partial) allocations, $\Alloc$ maximizes the vector $\vb{s}_{\Alloc} = (s_1, s_2, \ldots, s_n)$ lexicographically, i.e., it maximizes $s_1$, subject to that $s_2$, and so on.
\end{definition}
Note that a leximin allocation is guaranteed to exist. Also, if there are multiple leximin allocations $\Alloc$, then all of them have the sam sorted vector $\vb{s}_{\Alloc}$. Hence, all leximin allocations have the same Nash social welfare. 

The following theorems are known about leximin allocations.

\begin{theorem}[\cite{benabbou2020finding}]
\label{theorem:MNWiffLeximin}
Under matroid-rank valuations, an (partial) allocation is leximin iff it maximizes Nash social welfare.
\end{theorem}

Recall that, if the agents' valuations are matroid-rank functions, then a  Lorenz dominating allocation is guaranteed to exist~\cite{babaioff2020fair}. 

\begin{theorem}[Proposition 5 in \cite{babaioff2020fair}]
\label{theorem:LorenzthenMNW}
Under matroid-rank valuations, every Lorenz dominating allocation is leximin and it maximizes Nash social welfare as well. 
\end{theorem}

Now, we will restate and prove Lemma \ref{lemma:MNWisLorenz}. Note that this lemma complements the two theorems mentioned above. 

\LemmaMNWisLorenz*
\begin{proof}
The reverse direction (i.e., Lorenz domination implies Nash optimality) follows directly from Theorem \ref{theorem:LorenzthenMNW}. 

Hence, it remains to address the forward direction (Nash optimality implies Lorenz domination). Here, assume, towards a contradiction, that an (partial) allocation $\Alloc$ is Nash optimal but is not Lorenz dominating. Since $\Alloc$ maximizes Nash social welfare, Theorem \ref{theorem:MNWiffLeximin} implies that $\Alloc$ is also leximin. Therefore, allocation $\Alloc$ is leximin but not Lorenz dominating.

Note that all leximin allocations $\mathcal{B}$ have the same vector $\vb{s}_{\mathcal{B}}$ and the notion of Lorenz domination depends only on the vector $\vb{s}_{\mathcal{B}}$. Therefore, if leximin allocation $\Alloc$ is not Lorenz dominating, then no other leximin allocation can be Lorenz dominating. 
This contradicts the fact that Lorenz dominating allocations are themselves leximin (Theorem \ref{theorem:LorenzthenMNW}). Hence, the lemma stands proved. 
\end{proof}

\section{Index-Oblivious Mechanism for $\EFone$}
\label{section:efone-index-ob}

Recall that, for matroid-rank valuations, the mechanism $\PE$ outputs $\EFone$ allocations and it is truthful as well as Pareto efficient \cite{babaioff2020fair}. The following theorem shows that $\PE$ is index-oblivious as well. 

Notably, the proof of this theorem generalizes to show index-obliviousness for all deterministic mechanisms that output (partial) allocations $(A_1,\ldots A_n) $ solely  based on the their value tuple $(v_1(A_1), \ldots, \allowbreak v_n(A_n))$.

\begin{theorem}
\label{theorem:mnwio}
The $\PE$ mechanism is index-oblivious.
\end{theorem}
\begin{proof}
Let $\vb{v} = (v_1, v_2, \ldots, v_n)$ be any valuation profile and $\pi : [m] \mapsto [m]$ be any permutation. Write profile $\vb{v}^\pi = (v^\pi_1, \ldots, v^\pi_n)$.  along with allocations $(A_1, \ldots, A_n) = \PE(\vb{v})$ and $(A'_1, \ldots, A'_n) = \PE(\vb{v}^\pi)$. We will show that the non-wasteful allocations $(A_1,\ldots, A_n)$ and $(A'_1, \ldots A'_n)$ satisfy $v_i(A_i) = v^\pi_i(A'_i)$, for each $i \in [n]$, and, hence, $\PE$ is index-oblivious.

Recall that, given a valuation profile $\vb{v}$, mechanism $\PE$ considers the set all non-wasteful Lorenz dominating allocations and among them selects one, $(A_1, \ldots, A_n)$, that lexicographically maximizes $\left(v_1(A_1), v_2(A_2), \ldots, v_n(A_n) \right)$. 

Let $\mathcal{N}_{\vb{v}}$ denote the set of all non-wasteful allocations with respect to profile $\vb{v}$ and $\mathcal{T}_{\vb{v}} \subseteq \mathbb{R}_{\geq 0}^n$ be the corresponding set of valuation tuples, $\mathcal{T}_{\vb{v}} \coloneqq \left\{ \left( v_1(B_1), \ldots, v_n(B_n) \right) \ : \ (B_1, \ldots, B_n) \in \mathcal{N}_{\vb{v}} \right\}$.  
Furthermore, write $\mathcal{N}_{\vb{v}^\pi}$ and $\mathcal{T}_{\vb{v}^\pi}$ be the analogous sets for profile $\vb{v}^\pi$. 

To establish the result we will construct a bijection $\varphi$ between $\mathcal{N}_{\vb{v}}$ and $\mathcal{N}_{\vb{v}^\pi}$. This bijection $\varphi$ will satisfy the property that for every $\mathcal{B} = (B_1, \ldots, B_n) \in \mathcal{N}_{\vb{v}}$ and $(B'_1, \ldots, B'_n) = \varphi(\mathcal{B}) $ we have $v_i(B_i) = v^\pi_i(B'_i)$, for all $i$. Therefore, the existence of such a  bijection shows that $\mathcal{T}_{\vb{v}} = \mathcal{T}_{\vb{v}^\pi}$. Note that $\PE$ selects the optimal (with respect to the criteria mentioned above) tuple in $\mathcal{T}_{\vb{v}}$ for profile $\vb{v}$.  Analogously, $\PE$ selects the optimal tuple in $\mathcal{T}_{\vb{v}^\pi}$ for for profile $\vb{v}^\pi$. Since the sets are the same, we obtain the desired equalities $v_i(A_i) = v^\pi_i(A'_i)$  for all $i \in [n]$. This will establish the index-obliviousness of $\PE$.

Now, to complete the proof we will describe the bijection. For any (partial) allocation $\mathcal{B} = (B_1, \ldots, B_n)$, that is non-wasteful with respect to $\vb{v}$, we set $(B'_1, \ldots, B'_n) = \varphi(\mathcal{B})$ with $B'_i = \pi(B_i)$ for all $i \in [n]$. Since $\pi$ is a permutation of $[m]$, the constructed map $\varphi$ is a bijection. First, note that, with respect to $\vb{v}^\pi$, the (partial) allocation $(B'_1, \ldots, B'_n)$ is non-wasteful: $v^\pi_i(B'_i) = v_i \left( \pi^{-1} \left( B'_i \right) \right) = v_i \left( \pi^{-1} \left( \pi(B_i) \right) \right) = |B_i| = |B'_i|$; here, the first equality follows from the definition of $v^\pi_i$ and the second from the definition of $B'_i$. The third equality follows from the fact that  $\mathcal{B}$ is non-wasteful. We can analogously start with an (partial) allocation $(B'_1, B'_2, \ldots, B'_n)$, that is non-wasteful with respect to $\vb{v}^\pi$, and show that its pre-image under $\varphi$ is non-wasteful with respect to $\vb{v}$; specifically, $v_i(B_i) = v_i( \pi^{-1}(B'_i)) = v^\pi_i (B'_i) = |B'_i| = |B_i|$. Therefore, map $\varphi$ is a bijection between $\mathcal{N}_{\vb{v}}$ and $\mathcal{N}_{\vb{v}^\pi}$. 

Furthermore, 
\begin{align*}
v^\pi_i(B'_i) & = v_i(\pi^{-1}(B'_i)) \tag{by definition of $v^\pi_i$} \\
& = v_i(\pi^{-1}(\pi(B_i))) \tag{$B'_i = \pi(B_i)$} \\
& = v_i(B_i).
\end{align*}
Therefore, the valuation tuples are preserved under $\varphi$ and we get that $\mathcal{T}_{\vb{v}} = \mathcal{T}_{\vb{v}^\pi}$. As mentioned previously, this equality ensures that $\PE$ selects the same optimal value tuple under both the profiles $\vb{v}$ and $\vb{v}^\pi$, i.e., $v_i(A_i) = v^\pi_i(A'_i)$ for all $i$. The theorem stands proved.  
\end{proof}

\end{document}